\documentclass[11pt]{article}
\usepackage{fullpage}
\usepackage{amsmath,amsthm,amssymb}
\usepackage{xcolor}
\usepackage{hyperref}
\usepackage[nameinlink,capitalise]{cleveref}
\usepackage{bbold}
\usepackage{mathtools}
\usepackage{nicefrac}
\usepackage{tikz}

\makeatletter
\def\@fnsymbol#1{\ensuremath{\ifcase#1\or \dagger\or \ddagger\or
   \mathsection\or \mathparagraph\or \|\or **\or \dagger\dagger
   \or \ddagger\ddagger \else\@ctrerr\fi}}
\makeatother
\newcommand*\samethanks[1][\value{footnote}]{\footnotemark[#1]}

\DeclareMathOperator*{\argmax}{arg\,max}

\DeclareMathOperator*{\E}{\mathrm{E}}

\renewcommand{\d}{\,\mathrm{d}}

\newcommand{\rev}{\mathsf{Rev}}
\newcommand{\vcg}{\mathsf{VCG}}
\newcommand{\rvcg}{\mathsf{GTM}}
\newcommand{\myr}{\mathsf{Myerson}}
\newcommand{\spa}{\mathsf{SPA}}

\newcommand{\mec}{\mathsf{Mec}}
\newcommand{\mgtm}{\mathsf{MGTM}}
\let\sec\relax
\DeclareMathOperator*{\sec}{\mathrm{sec}}

\newtheorem{theorem}{Theorem}[section]
\newtheorem{lemma}[theorem]{Lemma}
\newtheorem{corollary}[theorem]{Corollary}

\theoremstyle{definition}

\newtheorem{definition}[theorem]{Definition}

\definecolor{myblue}{rgb}{0.25, 0.2, 0.75}
\definecolor{mygreen}{rgb}{0.15, 0.65, 0.25}
\definecolor{myred}{rgb}{0.75, 0.25, 0.15}
\hypersetup{
    colorlinks=true,
    linkcolor=myred,
    citecolor = myblue,      
    urlcolor=cyan,
}

\newtheorem*{theorem*}{Theorem}

\title{Prior-Independent Auctions for Heterogeneous Bidders}
\author{Guru Guruganesh\thanks{Google Research. Email: \texttt{\{gurug,aranyak,wadi\}@google.com}.} \and Aranyak Mehta\samethanks[1] \and Di Wang\samethanks[1] \and Kangning Wang\thanks{Stanford University. Email: \texttt{knwang@stanford.edu}.}}
\date{}

\begin{document}
\maketitle

\thispagestyle{empty}

\begin{abstract}
We study the design of prior-independent auctions in a setting with heterogeneous bidders. In particular, we consider the setting of selling to $n$ bidders whose values are drawn from $n$ independent but not necessarily identical distributions. We work in the robust auction design regime, where we assume the seller has no knowledge of the bidders' value distributions and must design a mechanism that is prior-independent. While there have been many strong results on prior-independent auction design in the i.i.d.\@ setting, not much is known for the heterogeneous setting, even though the latter is of significant practical importance. Unfortunately, no prior-independent mechanism can hope to always guarantee any approximation to Myerson's revenue in the heterogeneous setting; similarly, no prior-independent mechanism can consistently do better than the second-price auction. In light of this,  we design a family of (parametrized) randomized auctions which approximates at least one of these benchmarks: For heterogeneous bidders with regular value distributions, our mechanisms either achieve a good approximation of the expected revenue of an optimal mechanism (which knows the bidders' distributions) or exceeds that of the second-price auction by a certain multiplicative factor. The factor in the latter case naturally trades off with the approximation ratio of the former case. We show that our mechanism is optimal for such a trade-off between the two cases by establishing a matching lower bound. Our result extends to selling $k$ identical items to heterogeneous bidders with an additional $O\big(\ln^2 k\big)$-factor in our trade-off between the two cases.
\end{abstract}

%%%%%%%%%%%%%%%%%%%%%%%%%%%%%%%%%%%%%%%

\newpage
\clearpage
\setcounter{page}{1}
\section{Introduction}
Auctions are a fundamental component of online commerce and are a suitable mechanism for many different applications. The classic study of auctions focuses on two fundamental objectives: revenue maximization and welfare maximization. In the latter case, there is a long established theory starting with the seminal work of \cite{vickrey1961counterspeculation}
which established the optimal auction for maximizing the welfare of the agents.
The auction is relatively simple when we are selling a single item as it coincides with the second-price auction. It is desirable as it is very easy to implement in practice and it requires no assumption on the value distributions of the buyers.  

The main drawback of the classical second-price auction is that it achieves no guarantees in general with respect to the optimal revenue that a seller can gain. The optimal revenue is obtained by the celebrated Myerson's auction~\cite{myerson1981optimal} which uses a more intricate mechanism that utilizes knowledge of a prior distribution on the private values of the buyers. However, such distributional information is often hard to come by and therefore not always practically justified. While one could argue that the seller can learn the distribution over repeated auctions, in practice, distributions need not remain static over time. Furthermore, any such learning mechanism introduces new incentives for the bidders to manipulate the learning itself, making such mechanisms hard to analyze. All this makes the use of Myerson's auction impractical and the second-price auction more appealing despite its poor guarantees on revenue. The need to find mechanisms that do not rely on priors is often referred to as the ``Wilson doctrine'' or the ``Wilson critique''~\cite{wilson1989game} and has been discussed in several works related to ours, e.g.,~\cite{allouah2020prior,DBLP:journals/geb/DhangwatnotaiRY15}.
%\ggnote{ Addtionally, in the case of online e-commerce, the same auction must often be fixed a priori, and will be used in a wide variety of different settings where the bidders distributions can vary widely} \amnote{keep previous sentence? one can always do optimization of reserve across settings..}. 
%\ggnote{I think we can remove it since that it is not adding value but reserve optimization is not necessarily incentive compatible and hard to do in practice.}

%However, this creates two new issues: Firstly, in repeated auctions, the distributions themselves may not be static and could change over time. Secondly, any mechanism used by the seller to learn the distribution will create further incentives for the bidders to strategize to fool the learning mechanism and hence complicate the optimal bidding strategy for the bidders.

%This leads to the fundamental question of whether an auction which does not rely on priors can achieve good revenue guarantees.

%A natural definition of a ``simple'' auction is one that does not rely on priors; note that such an auction is necessarily also robust to changes in distributions.
%Nonetheless a fundamental open question remains whether a mechanism can achieve good revenue guarantees in a \textcolor{red}{setting without priors.}prior-free setting. 

A long line of work has tried to bridge the gap between these two different goals of maximizing revenue and not relying on prior information. Towards this, a canonical definition is that of \emph{prior-independent} auctions first introduced by \cite{DBLP:journals/geb/DhangwatnotaiRY15}.
Here, the goal is to find a mechanism -- typically dominant-strategy incentive-compatible (DSIC) -- that has no prior information but, no matter what the underlying value distribution is, has revenue competitive with the revenue-optimal auction tailored for that distribution. Specifically, in an i.i.d.\@ setting, once the auction is chosen, an adversary may choose any distribution (from a restricted class) and give each of the buyers a value drawn from this distribution. The mechanism is evaluated in expectation against the revenue achievable by the optimal mechanism that knows the distribution beforehand. In a series of results, beginning with the early work of~\cite{neeman2003effectiveness}, it is shown that if the distribution comes from a well-formed family such as a \emph{monotone-hazard-rate} (MHR) or \emph{regular} distribution, then it is possible to design such prior-independent auctions in the i.i.d.\@ setting. (We give a more detailed survey in~\cref{sec:related}.) The work of \cite{allouah2020prior} provided a characterization of optimal prior-independent mechanisms in the i.i.d.\@ setting, gave close upper and lower bounds for two-bidder regular distributions, and proved optimality of the second-price auction for two-bidder MHR distributions. \cite{DBLP:conf/focs/HartlineJL20} finally closed the gap by giving the tight bounds for two-bidder i.i.d.\@ regular distributions.
%\ggnote{Some of this paragraph seems to belong in the related works setting. I have moved it. PTAL }
%51.9 vs 55.59. for regular, and SPA is optimal for MHR at 71.xx

While the i.i.d.\@ model for the priors is amenable to clean results and provides nice mathematical intuition, in many practical applications, we do not expect to see identical bidders.  In this paper, we investigate the potential of prior-independent auction design for revenue maximization for independent but non-identical (i.e.\@ heterogeneous) bidders. 
%This model is important to study from a practitioner's view, since practical auctions often include bidders whose value distributions have vastly different characteristics. 
This is especially important in online ad auctions where in the same auction we see advertisers of different scales, or with different goals such as brand advertising, targeted advertising and performance-driven advertising (see, e.g., \cite{golrezai2021boosted}). It is therefore imperative to examine what kind of results are achievable in the prior-independent setting with non-identical bidders.

\subsection{Our Results and Techniques}
\subsubsection*{Metric and Benchmarks}
Extending these results to the heterogeneous\footnote{We use the term ``heterogeneous'' to stand for independent but not necessarily identical distributions.} setting immediately runs into a challenge: the lack of an existing benchmark that is effective at distinguishing the performance of mechanisms in our context. Indeed, as long as the family of distributions being considered is reasonably general (e.g.\@ MHR or regular distributions), no prior-independent mechanism can achieve anything non-trivial with respect to the canonical benchmarks.

%apart  that is both feasible and nontrivial in order to measure the performance of a given mechanism  %One immediate question is how to measure the performance of a given mechanism in the prior-independent setting. 
\begin{itemize}
    \item \textbf{No mechanism can guarantee any $\varepsilon>0$ approximation of the Myerson revenue.} 
    %Consider the canonical benchmark of the revenue of the optimal auction that knows the bidder distributions, which is the benchmark used in the i.i.d.\@ case.
    %We first note that it is impossible to achieve any non-trivial approximation to the Myerson benchmark in the heterogeneous setting in general.
    Consider a simple example of two buyers with deterministic (but unknown) values $v_1 = 1$ and $v_2 = x > 1$, where the Myerson revenue is $x$. Informally speaking,\footnote{\cref{thm:lower} with $\tau \to +\infty$ gives a rigorous proof. Also note one can turn any deterministic value in our examples into a uniform distribution over a tiny range around that value to have continuous and regular bidder value distribution.} if for any $\varepsilon>0$ a DSIC prior-independent mechanism $M$ can guarantee at least $\varepsilon \cdot x$ for all $x$'s, then the probability of allocating to buyer $2$ must strictly increase with $x$ (by DSIC) and eventually has to go above $1$, which gives a contradiction. %it is straightforward to see by induction starting from any $x=x_0\geq 1/\varepsilon$ that whenever $x$ grows by a factor of $1/\varepsilon$, the probability of $M$ selling to the second bidder must increase by at least $\Omega(\varepsilon)$, which leads to a contradiction since that probability cannot grow larger than $1$. 
    
    Noting that revenue maximization is trivial for multiple i.i.d.\@ point distributions (via e.g.\@ a second-price auction), we can see that this example already highlights the comparative difficulty in the non-i.i.d.\@ setting. This emphasizes that the prior-independent results in the i.i.d.\@ setting cited above heavily leverage the identical nature of the distributions, originating from the intuition from the work of \cite{bulow1996auctions} that an additional i.i.d.\@ bidder's random value draw serves as a reasonably good reserve price (see also~\cite{DBLP:journals/geb/DhangwatnotaiRY15}). 
%Although the above example only uses point distributions (which are trivial in the i.i.d setting), show the difficulty of competing with the optimal auction in the non-i.i.d section. %\amnote{(note that for a deterministic values case as above, the iid setting is trivial and a second-price auction has optimal revenue)}.
%that an additional identical bidder's bid provides a good estimate to the optimal reserve price for the distribution\amnote{check}. 
\item \textbf{No mechanism can guarantee at least  $1+\varepsilon$ times the second-price for any $\varepsilon>0$. } With the above observation, it is natural to turn to the canonical benchmark in the prior-free setting, and ask if we can always beat the expected revenue of the second-price auction, say by a $(1+\varepsilon)$ factor for some $\varepsilon>0$. However, this is not possible even over MHR distributions, since it is shown in \cite{allouah2020prior} that the second-price auction is the optimal prior-independent mechanism for i.i.d.\@ MHR distributions with two bidders.\footnote{Our impossibility results (\cref{thm:ab_lower2}) also rule out mechanisms with ``beating second-price'' type guarantees such as getting at least $(1+\varepsilon)$ of the second-price revenue whenever possible, and otherwise (i.e. on distributions where Myerson gets less than $(1+\varepsilon)$ times the second-price revenue) getting the same as second-price revenue.}
\end{itemize}

This suggests that in the heterogeneous context, any non-trivial guarantee w.r.t. either benchmark would be infeasible, and thus we need to find a new benchmark to tell apart good and bad mechanisms (in terms of revenue guarantees). For this purpose, it is illustrative to revisit our earlier two-buyer example with the family of point distributions where $v_1=1$ and $v_2=x$ for all $x \geq 1$. Consider any DSIC prior-independent mechanism $M$, and denote $M(x)$ as the expected revenue of $M$ when $v_2=x$. Myerson knows the distributions and thus always gets revenue $x$; the second-price always gets revenue $1$; and intuitively $M(x)$ should look like \cref{fig:rev_curve}.\footnote{The figure is only for illustration purposes and not meant to be rigorous. Note that $x$ exactly captures the how non-identical the bidders are in this example, and intuitively one can imagine a similar figure in general where the $x$-axis captures the level of heterogeneity of the bidders, although coming up with a formal metric of it is beyond the scope of our study.}
%
%\footnote{The discussion is only for intuitions and not rigorous, since mechanisms based on correctly guessing the fixed $v_1$ break our argument. %Technically, if $v_1$ is always $1$ (or any fixed constant) in all the distributions being considered, mechanisms can guess the constant $v_1$ and if guessed correctly can exploit that knowledge and our argument won't hold. 
%The rigorous version is to consider the family of $v_1=a,v_2=ax$ or $v_1=ax,v_2=a$ for all $a>0,x\geq 1$, %let $M(a,x)$ be the expected revenue of M for any particular values of $a,x$, 
%and let $M(x)=\min_a{M(a,x)/a}$. In general, one can restrict to scale-free and symmetric mechanisms (see \cref{sec:characterization}), so the ratio between the two buyers (i.e., $x$) captures the essence of the rigorous example.}
%\footnote{Note $x$ exactly captures the how non-identical the bidders are in this example, and intuitively one can imagine a similar figure in general where the $x$-axis captures the level of heterogeneity of the bidders, although coming up with a formal metric of it is beyond the scope of our study.}  We can wlog assume $M(x)$ to be non-decreasing with $x$.\footnote{Otherwise one can easily get another mechanism $M'$ such that $M'(x)$ is non-decreasing and $M'(x)\geq M(x)$ for all $x$.} 
%Moreover, if $M(x)>1$ for some $x$, $M$ must with some probability allocate the item to buyer $2$ and that probability should also be non-decreasing with $x$ ($M$ is DSIC). 
Firstly, as discussed earlier, we know for larger $x$'s eventually we won't be able to compare with $x$, so we can only look at how much better we do compared to $1$, which is the only other quantity in the system. In addition, to achieve higher revenue than SPA for larger $x$'s, a DSIC mechanism must reduce the probability of allocating to buyer $2$ for smaller $x$'s, which means larger gap compared to the optimal revenue $x$ for smaller $x$'s (as the price for buyer $1$ is at most $1\leq x$).

\definecolor{c2}{rgb}{0.25, 0.2, 0.75}
\definecolor{c1}{rgb}{0.15, 0.65, 0.25}
\begin{figure}[h]
\centering
\begin{tikzpicture}[scale=0.8]
  \draw [->] (-0.2, 0) to (6.5, 0);
  \draw [->] (0, -0.2) to (0, 6.5);
  \draw [dotted] (1, 0) to (1, 6.5);
  \draw [dotted] (0, 1) to (1, 1);
  \draw [very thick, c1] (1, 1) to (6, 1);
  \draw [very thick, c1] (1, 1) to (6, 6);
  \draw [very thick, c2, domain=1:6] plot (\x, {0.4+0.4*ln(\x)+0.1*\x});
  \node at (-0.25, -0.25) {$0$};
  \node at (-0.25, 1) {$1$};
  \node at (1, -0.25) {$1$};
  \node at (6.5, -0.25) {$x$};
  \node at (-0.2, 6.8) {Revenue};
  \node at (6.5, 1) {\color{c1} $\spa$};
  \node at (6.9, 6) {\color{c1} $\myr$};
  \node at (6.6, 1.7) {\color{c2} $M(x)$};
\end{tikzpicture}
\caption{revenue curves for different point distributions in the illustrative example}
\label{fig:rev_curve}
\end{figure}
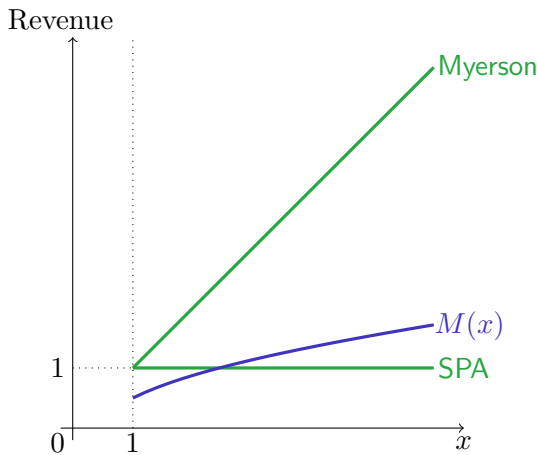

It is thus clear in this example that when considering worst-case guarantees (i.e. hold for all $x$), it is a necessary trade-off for \emph{any} prior-independent mechanism between how much advantage compared to $1$ for larger $x$'s and how much loss compared to $x$ for smaller $x$'s. A meaningful benchmark in more general cases must also capture this similar phenomenon, i.e. achieving higher revenue (compared to second-price) in high heterogeneity cases must sacrifice revenue (compared to Myerson) when heterogeneity is low (e.g., i.i.d.\@ case). %This naturally motivates the either-or type benchmark we use in our results. In particular, for those cases when it is impossible to beat second-price auction on revenue in our above discussion, though, we note that the second-price auction itself approximates the optimal revenue. 
This motivates a natural either-or type benchmark which integrates the two canonical benchmarks: on any bidder distributions (from some family of distributions) the expected revenue of the mechanism can either beat the second-price auction by some multiplicative factor $\alpha$, or is a $\beta$-approximation of the Myerson revenue. Our benchmark takes the either-or form because we do not have a good metric to capture the heterogeneity level in general (and not aim to come up with one in this work), and want a worst-case benchmark (i.e. holds for any level of heterogeneity). Informally $\alpha$ captures the revenue guarantee for higher heterogeneity (i.e. larger $x$'s in our example) and $\beta$ captures the revenue guarantee on the other end (i.e. smaller $x$'s). 

%{\color{blue} Do we still want to keep this part from before? A little inconsistent with the new motivation} This flexibility allows a mechanism to circumvent the bad distributions which prohibit any non-trivial approximation on one benchmark, while still giving a qualitatively interesting result (and in some sense the only thing one can hope for) in those cases by a good approximation guarantee on the other benchmark. As we will show, the bad events are mutually exclusive and that there is a prior-independent mechanism that either is good compared to Myerson's auction or is considerably better than the second-price auction.  % In light of these results, a natural question is if it is possible to achieve a good approximation with respect to either of these two canonical auctions, namely second-price auction and Myerson. 

% a metric is also useful to an auctioneer in practice, since it bases the comparison of the revenue of a mechanism to the revenue of the most canonical auction -- Vickery Auction -- and to the optimal auction (i.e., Myerson) if priors were known. While (as our impossibility results show) no prior-independent mechanism can always be competitive to either of the two, by tuning the parameter, our benchmark allows an auctioneer to smoothly interpolate the trade-off between the two across different instances. %In particular, our benchmarks are the expected revenue that the Myerson's auction would achieve if the distribution were known to the seller beforehand and the expected revenue achieved by the Vickrey Auction. 

\subsubsection*{Main results}
Our main result (\cref{sec:main-result}) is that one can design mechanisms that achieve either a constant fraction of the optimal Myerson's revenue or beat the revenue of the second-price auction by a constant factor. 
\begin{theorem*} [Informal]
For any parameter $\tau > e$ there exists a prior-independent randomized mechanism such that for $n$ buyers with private values drawn from any independent (but not necessarily identical) regular distributions, the expected revenue of the mechanism is at least
\[
\min\left(\Omega\Big(\frac{1}{\ln\tau}\Big) \cdot \myr, \ \tau \cdot \spa\right),
\]
where $\myr$ is the expected revenue of Myerson's auction for the given distributions, and $\spa$ is that of the second-price auction.
% For any $\tau \geq e$ there exists a prior-independent randomized mechanism that given $n$ buyers, drawn from independent (but not necessarily identical) regular distributions can either achieve a expected revenue that is at least a $\tau$ multiple of the expected revenue of the second-price auction or at least $\Omega(\frac{1}{\ln(\tau)})$ factor of the optimal expected revenue achieved by Myerson's auction. 
\end{theorem*}

%We note that such a metric is a useful benchmark to an auctioneer in practice. It shows that one can beat the revenue guarantees of the most canonical auction -- Vickery Auction or guarantee a constant fraction of the revenue-optimal auction. Our benchmark allows an auctioneer to compare themselves against a smooth interpolation between a simple auction (namely the Vickrey Auction) and a clairvoyant optimal auction (namely Myerson's auction). 

%\footnote{\color{blue} This is the original version: which perspective should we keep? 'there is an auction which approximates the optimal revenue in the good cases, but recovers elegantly in the bad cases: the expected revenue in such cases is not only bounded from below, but is in fact a $\tau$-factor better than that of the canonical second-price auction.}  

 %This captures (all?) the currently known (or expected to hold) cases\footnote{\color{blue}Is this true?} when approximating Myerson's revenue is possible, e.g.\@ the i.i.d.\@ MHR or regular distributions (and intuitively should extend to the case of non-i.i.d.\@ but not too different buyer distributions). 
We emphasize that the auctioneer can choose the value of the mechanism's parameter $\tau$, and the resulting revenue guarantee will hold  all distributions no matter the level of heterogeneity. One can interpret our benchmark as a \emph{robustness guarantee} in the sense that
 when SPA is guaranteed to be at least a (good) constant fraction of Myerson (e.g., in the i.i.d.\@ or nearly-i.i.d.\@ cases), the $\tau\cdot \spa$ part in the benchmark is a good approximation of Myerson, so a guarantee against the $\min$ still approximates the optimal revenue well. 
On the other hand, the benchmark still recovers elegantly in the cases when approximating Myerson is not possible: the expected revenue in such cases is not only bounded from below, but is in fact a $\tau$-factor better than that of SPA. 
%This also shows how such a benchmark is useful to an auctioneer in practice -- by choosing the value of the mechanism's parameter $\tau$, the auctioneer can choose a point in the tradeoff space, and the resulting revenue guarantee will hold no matter the level of heterogeneity of the distributions.

We complement our main result with an almost matching upper bound.%\amnote{How do you get rid of the theorem numbers in these informal theorems?}
\begin{theorem*}[Informal]
For $2$ bidders and any $\tau \geq 3$, there is no prior-independent mechanism with expected revenue greater than
\[
\min\left(\frac{2.5}{\ln \tau} \cdot \myr, \ \tau \cdot \spa\right),
\]
for all pairs of independent value distributions.
%either greater than a multiple $\tau$ of the expected revenue of the second-price auction, or greater than a factor $\frac{2.5}{\ln{\tau}}$ of the expected revenue of Myerson's auction.
\end{theorem*}
The upper bound above holds even over a very restricted class of distributions -- point distributions, i.e., when the two bidders values are deterministic.\footnote{This is in stark contrast to the case of point distributions in the i.i.d.\@ setting, where getting Myerson's revenue is easy by using e.g.\@ a second-price auction.} Note that these distributions are also MHR, thus showing that the main result above cannot be improved significantly for the class of MHR distributions.

To complete the picture, we show (in \cref{sec:impossibility}) that the main result is tight along different dimensions: One cannot achieve even a minor approximation of this form if either (a) the distributions are allowed to be general (non-regular) independent ones, or (b) the class of distributions is regular, but the mechanism is deterministic. Specifically, in these settings one cannot have a mechanism that achieves either a $\left(1+\varepsilon\right)$-multiple of the revenue of the second-price auction or an $\varepsilon$-factor of the revenue of Myerson's auction, for any constant $\varepsilon > 0$.

In \cref{sec:multi-item}, we extend our main result to the multiple-identical-item case. When we have $k$ identical items to sell and the buyers are unit-demand, we extend our mechanism and achieve a very similar trade-off between the two sides of the revenue objective, losing a factor of $O(\ln^2 k)$ in the trade-off.\footnote{This loss is $O(\ln k)$ when $\tau$ is large.} We prove:
\begin{theorem*}[Informal]
For $n$ buyers and $k$ items, there exists a prior-independent randomized mechanism achieving a revenue of at least
\[
\min\left(\Omega\Big(\frac{1}{\ln(k\tau)} \cdot \frac{1}{\ln k}\Big) \cdot \myr, \ \tau \cdot \vcg\right).
\]
\end{theorem*}

Finally, in \cref{sec:characterization} we provide a characterization of optimal prior-independent mechanisms in our heterogeneous setting for $2$ buyers (similar to the characterization for the i.i.d.\@ setting in~\cite{allouah2020prior}).

\subsubsection*{Techniques}
Our main result uses a simple class of randomized Threshold Mechanisms (see~\cref{sec:threshold} for definition). %which simultaneously provide a good approximation with respect to the optimal Myerson Auction or the second price auction. 
%These mechanisms use randomness to smoothly interpolate between the worst case instances for Second Price Auction and the Myerson Auction.  %\ggnote{Define threshold mechanisms in words.}
%In particular, we show that either the threshold mechanism can achieve a good guarantee compared to Myerson's auction. Otherwise, it must be the case that the highest bidder in the auction must be sufficiently higher than the remaining bids and thus we can extract more revenue from the bidder. \ggnote{We need to say more here.}
%At a high level, the following intuition suggests how such a result could be possible. In the case that the instance includes a run-away high distribution, an auctioneer should be able to extract revenue greater than the second-price auction. On the other hand, if the bidders' value distributions are close to each other, then the auctioneer should be able to approximate the optimal revenue in a prior-independent manner. 
%To be able to perform well against both these types of value distributions it is critical to use randomization. We show that a simple class of \emph{threshold mechanisms} (see~\cref{sec:threshold} for a formal definition) can perform well in both these scenarios. 
Informally, our threshold mechanism chooses a threshold from a carefully constructed distribution over thresholds and only allocate the item if the winner's bid is higher than the next highest bid by a factor of this threshold. It is not hard to see intuitively why one turns to this class of mechanisms in the prior-independent setting: Since a mechanism should be scale-free (see \cref{sec:characterization}), we exclude mechanisms such as posted-price or using reserve prices. Moreover, the mechanism has to be randomized, again by considering the simple two-bidder point distributions example from before. Our idea of randomizing from a set of geometrically separated thresholds can be found in many works in the field of approximation algorithms.

Our novel benchmark poses interesting new mathematical challenges: We emphasize that our benchmark is the minimum (which is a non-linear operator) over the expected revenues of the two canonical auctions over the distributions -- it is a much stronger requirement than beating the minimum of the two auctions in each realization of values. There have been related work on prior-free auctions (discussed in \cref{sec:related}) -- they can been seen as prior-independent auctions used on \emph{point} value distributions (i.e. the much weaker requirement mentioned in the preceding sentence) and hence are very special cases of our setting of \emph{regular} value distributions. (Recall that we made the regularity assumption since a similar result is impossible without any distributional assumption, according to \cref{thm:no_regularity}.) To achieve our results, we first use the median of the highest order statistic as an upper bound on the potential revenue achieved by Myerson's auction. We next show that the second-price auction's revenue is related to lower order statistics. We then argue that, depending on the gap between the order statistics, our threshold mechanism can either approximate the highest order statistic or beat the lower order statistics by a suitable margin. This relies on the randomness in our mechanism as it is oblivious to the order-statistics information. The main challenge and novelty in the proof is the \emph{distributional} analysis for arbitrary regular distributions -- existing works on prior-free auctions, as the name suggests, do not share with our work this challenge. We believe our techniques have potential applications in other settings related to regular distributions.

As we extend the result to the multi-item setting, these techniques do not suffice. We need to be competitive over a larger set of possibilities depending on the particular decomposition of the distributions. A simple extension would lose a factor of $k$ (the number of items) in the trade-off. To hedge against all these instances, we introduce the technique of randomly limiting the number of items sold, in a manner that does not lose too much in expected revenue. On the one hand, we could sell a large number of items at a small price or sell a few items for a large price. By randomly choosing the number, we hope to set the correct price more often. One particular challenge that arises is that to bound expected revenue of the optimal Myerson auction is non-trivial here. We use a subtle definition of the appropriate order statistics of a carefully chosen subset of the bidders to argue that we can still achieve a constant fraction of Myerson or beat the Vickrey Auction appropriately. This helps us reduce the loss in the trade-off to only $O(\ln^2 k)$. Whether  this can be reduced to a constant is an intriguing open question.   

\subsection{Related Work} \label{sec:related}

There is a lot of related work in the area related to revenue maximization starting with the work of \cite{myerson1981optimal}. Due to the practical and technical difficulty of dealing with priors as mentioned earlier, there has been a great deal of work in coming up with prior-independent mechanisms. These works can be classified into a few major related strands.

The first major strand is when the auction is standard and simple, e.g.\@ Vickrey, but the mechanism tries to recruit additional bidders to ensure that the new instance can compete with respect to the optimal revenue. This line of work was initiated by \cite{bulow1996auctions}. This work gained a great deal of interest in the algorithmic game theory community through the result of \cite{HR2009simple} who showed revenue guarantees for the Vickrey Auction with additional bidders. In particular, they showed that in quite general settings, a VCG Auction with $n$ additional bidders can compete with the optimal revenue when the bidders are drawn from heterogeneous regular distributions. This was extended to a more general class of distributions in the work of \cite{sivan2013vickrey}. The best results along these lines were obtained by \cite{fu2019vickrey} who showed that by carefully choosing the extra bidders, one needs only one extra bidder for the single-item case. A number of recent works (see e.g.\@ \cite{feldman201899,eden2016competition,beyhaghi19competition,DBLP:conf/sigecom/CaiS21}) showed that recruiting additional bidders provides a simple set of mechanisms that achieve
near optimal mechanisms even with multiple items and multi-dimensional valuations. Note that in the latter case, the optimal mechanisms are known to be extremely complex. All of the above results argue that a certain level of additional bidders allows a simple mechanism such as the Vickery Auction to compete with the revenue-optimal mechanism. 
%These results can be viewed by a auctioneer as suggesting that additional efforts to raise revenue should be directed towards recruiting new bidders rather than modifying the mechanism. 
However, it may not be possible to recruit additional bidders because they may not exist or may come at a great cost. Furthermore, these results offer no guarantees with respect to the instance with those additional bidders.

Another line of work attempts to produce bounds on the approximation factor of the revenue of certain simple mechanisms, potentially with knowledge of the prior, compared to the optimal revenue on the same instance. Notably, the bidder-augmentation result of \cite{bulow1996auctions} can be reinterpreted (\cite{DBLP:journals/geb/DhangwatnotaiRY15}) to show that the expected revenue of the Vickrey Auction is at least $\nicefrac{(n-1)}{n} \geq \nicefrac{1}{2}$ of the expected revenue of Myerson's auction for i.i.d.\@ distributions. A series of works, starting with~\cite{HR2009simple}, study \emph{simple} auctions which are competitive to the optimal auction even in heterogeneous settings. These results, including~\cite{alaei2019optimal,jin19pricing}, consider the competitiveness of simple auctions such as the second-price auction with an anonymous reserve. Note that these latter auctions, although simple, still depend on the knowledge of the prior.

Recently, there has been a third line of research building on and improving the result that Vickrey is a $\nicefrac{1}{2}$-approximation in the i.i.d.\@ setting, via new prior-independent auctions in two papers highly relevant to our work. Firstly, \cite{fu2015randomization} showed that in the case of a single-item and i.i.d.\@ bidders with a regular distribution, one can beat the revenue approximation guarantees of the Vickrey Auction through the use of randomization. \cite{fu2015randomization} introduced a randomized prior-independent auction called $(\varepsilon, \delta)$-inflated second-price auction: with probability $\varepsilon$ it runs a second-price auction, and with the rest of the probability the highest bidder wins only if its bid is greater than the next highest bid by a factor of $\alpha \geq 1$; otherwise the item is unallocated. They proved that this auction achieves a fraction of the optimal revenue strictly larger than $\frac{n-1}{n}$ fraction (for $n\geq 2$ bidders), and an improved factor of $0.512$ for two bidders, thus beating the approximation guarantee of the second-price auction. This result was further greatly generalized in~\cite{allouah2020prior} which introduced a family of prior-independent auctions called \emph{threshold-auctions}, and proved stronger results for revenue in the prior-independent two bidder setting. They show an approximation factor of $0.715$ for i.i.d.\@ MHR and $0.519$ for i.i.d.\@ regular distributions. The factor for MHR distributions is achieved by the second-price auction and is shown to be optimal, while the factor for regular distributions is achieved by a new auction in the class of threshold mechanisms; an upper bound of $0.556$ is also shown (under a technical assumption of finite Arzel\`a variation). This gap between $0.519$ and $0.556$ was finally closed (under the same technical assumption) by the work of \cite{DBLP:conf/focs/HartlineJL20}, who showed that the optimal prior-independent auction gives an approximation factor of $0.524$ for two bidders with i.i.d.\@ regular distributions. Our paper lies in this thread of research and studies the non-identical setting which is arguably more relevant practically as discussed earlier. We note that while some prior work, e.g., \cite{DBLP:journals/geb/DhangwatnotaiRY15}, does consider prior-independent approximation in a heterogeneous setting, such results still need to assume the existence of multiple bidders with the same attribute, i.e., distribution, which can essentially serve as i.i.d.\@ replacements for each other. We note that our auction $\rvcg$ also lies in the family of threshold-auctions introduced in \cite{allouah2020prior}. This family has been further studied (see e.g.\@ \cite{Mehta22,DBLP:conf/www/LiawMP23}) to show stronger welfare guarantees beyond VCG in the auto-bidding setting, which is an increasingly important area in the online advertising industry. This suggests our results on revenue guarantees may be of significant practical interest.

Another way to deal with the distributional assumption is to understand the cost of learning the prior distribution from repeated auctions. \cite{kleinberg2003value} considered the case where one must learn the value of the buyers' distributions using posted-price mechanisms. \cite{cole2014sample,DBLP:conf/stoc/GuoHZ19} considered the question of determining how many samples one needs from a distribution to compute the optimal mechanism for revenue maximization. There is a line of work on approximately revenue-optimal auctions with access of $1$ sample (see e.g.\@ \cite{DBLP:journals/geb/DhangwatnotaiRY15,DBLP:conf/soda/AzarKW14,DBLP:conf/ec/CorreaDFS19}).

Early works on the closely related direction of prior-free auctions include \cite{DBLP:conf/soda/GoldbergHW01,DBLP:journals/geb/GoldbergHKSW06,DBLP:conf/stoc/ChenGL14}, where the guarantees are worst-case and valuations are not even drawn from prior distributions. From another point of view, prior-free settings are prior-independent settings limited to (heterogeneous) point distributions. There are other algorithms that study the prior-free setting with additional assumptions such as the buyers having a specific form as such following a low-regret algorithm or participating in a dynamic auction where the state changes as the auctioneer must have limited liability (see e.g.\@ \cite{deng2019prior,braverman2021prior}).
Another approach towards robust auction design is that of distributionally-robust auctions which assumes that the auctioneer has knowledge of some summary statistics of the distribution such as the mean and the upper limit of the support, and characterizes the max-min performance, i.e., under the worst case distribution (see \cite{bachrach2022distributional,che2019distributionally}. A recent work of \cite{anunrojwong2022robustness} also tackles the question of designing optimal mechanisms for prior-independent distributions but considers the benchmark of regret. Their results focus on categorizing additive loss between the best mechanism and the optimal welfare that can be achieved and do not translate to giving multiplicative approximations as in our work. 
Recently, there is another line of work on ``revelation gaps'' that studies non-truthful auctions in the prior-independent framework~\cite{DBLP:conf/focs/FengH18,DBLP:conf/stoc/FengHL21}. \cite{hartline2021lower} give lower bounds on prior-independent auctions in the i.i.d.\@ setting.

\section{Preliminaries}
We consider the setting of selling one indivisible item to $n$ buyers. Buyer $i$ has valuation $v_i$ for the item. $v_i$ is drawn from a distribution $V_i$, and they are mutually independent. Different from a classical setting where $V_i$'s are public information, we assume the seller and the buyers do not have access to these distributions. The seller, therefore, must use a \emph{prior-independent} auction to sell the item. Her goal is to maximize her expected revenue, using a direct, dominant-strategy incentive-compatible (DSIC) mechanism (meaning that the mechanism asks each buyer for their valuation, and for each buyer, reporting their true valuation is a dominant strategy).

We use $\myr$ to denote the revenue-optimal auction characterized by the seminal work of \cite{myerson1981optimal}, and use $\spa$ to denote the second-price auction. Notice that $\myr$ is not prior-independent while $\spa$ is. If the context is clear, we also use $\myr$ and $\spa$ to denote their respective expected revenue on some given instance.

We use $v^{(k)}$ to denote the $k$-th maximum value in $\{v_1, v_2, \ldots, v_n\}$. In particular, $v^{(1)}$ is the maximum value and $\spa = \E_{(v_1, \ldots, v_n) \sim (V_1, \ldots, V_n)} \left[v^{(2)}\right]$. Moreover, we use $s_k$ to denote the median of the distribution of $v^{(k)}$.

\subsection{Regularity}
Many of our results use the notion of \emph{regularity} on the value distributions. The regularity assumption is frequently imposed in the literature of auction design, ever since the seminal work of \cite{myerson1981optimal}.

In the proof of our positive results, we restrict regular distributions to be continuous (which is commonly assumed). This allows us to define the median $s$ of the distribution for some random value $v$ to satisfy $\Pr[v \geq s] = 1/2$, and similarly for other quantiles, to simplify the exposition. The proofs directly generalize to point distributions (i.e. deterministic values; they can be considered as more broadly construed ``regular distributions''). However, in our negative results, we consider point distributions (i.e. deterministic values) to be regular.\footnote{In the negative results of \cref{thm:no_regularity}, \cref{thm:lower}, \cref{thm:ab_lower1} and \cref{thm:ab_lower2}, we can slightly perturb the distributions to make them continuous regular, without changing the message of the proof. \cref{thm:no_randomness} does utilize the fact that there can be a probability mass.}

Formally, regular distributions are the distributions where the virtual value function $\varphi(v) := v - \frac{1 - F(v)}{f(v)}$ is nondecreasing in $v$, where $f(\cdot)$ and $F(\cdot)$ are the corresponding probability density function (PDF) and cumulative distribution function (CDF). Special cases of regular distributions include all monotone-hazard-rate (MHR) distributions -- the distributions with hazard rate $\frac{f(v)}{1 - F(v)}$ nondecreasing in $v$.

%Throughout our discussion, we restrict regular distributions to be continuous (which is commonly assumed).\footnote{Regular distributions have monotone virtual values, which precludes any point mass except at the bottom point of the support. Our results easily generalize to this slightly more general setting.} This allows us to define the median $s$ of the distribution for some random value $v$ to satisfy $\Pr[v \geq s] = 1/2$, and similarly for other quantiles. However, in our negative results, for the purpose of simpler description, we consider point distributions (i.e. deterministic values) to be regular. They can be perturbed to construct continuous regular distributions for the same negative results.

%This exclude the corner case of point distributions (i.e. deterministic values), but this is without loss of generality as we can use a continuous distribution tightly concentrated around the point value without changing the revenue of any mechanism considered in our paper.  

\cref{lem:regularity} states a property of a regular distribution, which we will use later.
\begin{lemma}
For $v$ drawn from a regular distribution $V$, let $r$ be its Myerson's reserve (i.e., $r \in \argmax_{p} p \cdot \Pr[v \geq p]$), and $s$ be its median (i.e. $\Pr[v \geq s] = \frac{1}{2}$). We have
\begin{enumerate}
    \item $r \cdot \Pr[v \geq r] \leq s.$ In other words, Myerson's revenue is at most $s$, and thus at most twice the revenue of selling at $s$.
    \item If $s \leq \ell \leq r$ for some $\ell$, then $r \cdot \Pr[v \geq r] \leq 2 \cdot \ell \cdot \Pr[v \geq \ell]$.
\end{enumerate}
 
\label{lem:regularity}
\end{lemma}
\begin{proof}
If $r \leq s$, then clearly $r \cdot \Pr[v \geq r] \leq r \leq s$. Otherwise, let $q(p) = \Pr[v \geq p]$ and consider any $\ell \in [s, r]$. For a regular distribution, the revenue is concave in the quantile space, i.e., $v \cdot q(v)$ is concave in $q(v)$. Therefore,
\[
\ell \cdot q(\ell) \geq 0 \cdot q(0) \cdot \frac{q(\ell) - q(r)}{q(0) - q(r)} + r \cdot q(r) \cdot \frac{q(0) - q(\ell)}{q(0) - q(r)} = r \cdot q(r) \cdot \frac{1 - q(\ell)}{1 - q(r)} \geq \frac{1}{2} \cdot r \cdot q(r)
\]
and therefore proves Statement 2.

Noticing that $q(s) = \frac{1}{2}$, we have $s \geq r \cdot q(r)$ if we set $\ell = s$, which proves Statement 1.
\end{proof}

\subsection{The Threshold Mechanisms}
\label{sec:threshold}
We will use a class of prior-independent mechanisms: the \emph{threshold mechanisms}. A threshold mechanism uses a finite number of thresholds $\{\lambda_1, \lambda_2, \ldots, \lambda_m\}$, where $\lambda_i$ happens with probability $w_i$, with $\sum_{i = 1}^m w_i = 1$. For a value profile $(v_1, v_2, \ldots, v_n)$, the mechanism generates a random threshold $\lambda_i$ according to the probabilities $(w_1, w_2, \ldots, w_m)$. It then looks at the top two values, $v^{(1)}$ and $v^{(2)}$. If $v^{(1)} \geq \lambda_i \cdot v^{(2)}$, the item is allocated to the buyer with the highest value for a price of $\lambda_i \cdot v^{(2)}$.  Otherwise, the item is not allocated.
(Since we are considering continuous distributions, ties happen with probability $0$, and we omit the mechanism's behavior to resolve ties. Ties can be broken in any consistent way without affecting our results.) Note that the threshold mechanisms are dominant-strategy incentive-compatible (DSIC).

We also generalize this definition to the multi-item case, which can be found in \cref{sec:multi-item}.

\section{Impossibility Results} \label{sec:impossibility}

As we noted in the introduction, our goal is to obtain a mechanism that is competitive with $\myr$ or significantly beats $\spa$ in a prior-independent setting. However, if distributions can be different across all the bidders, we will need to make some assumptions on the class of distributions of the buyers. As we show in \cref{thm:no_regularity}, no such mechanism can achieve the desired form of guarantees for all general distributions. (Following the convention of this line of work on prior-independent auctions, we limit our attention to DSIC mechanisms.)

\begin{theorem}
For any constant $\varepsilon > 0$, no DSIC mechanism can guarantee a revenue of $\min(\varepsilon \cdot \myr, (1 + \varepsilon) \cdot \spa)$ for general valuation distributions, even when there are only two buyers.
\label{thm:no_regularity}
\end{theorem}
\begin{proof}
Fix an integer $k > 1$, and define for $j = 1,\ldots, k$, the instance $\mathcal{I}_j$ to be the following:
\begin{itemize}
\item $v_1 = \sqrt{k} \cdot 2^j$ with probability $2^{-j}$; and $v_1 = 1$ with probability $1 - 2^{-j}$.
\item $v_2 = 1$ with probability $1$.
\end{itemize}
We have $\spa(\mathcal{I}_j) = 1$, since $v^{(2)} = 1$; and $\myr(\mathcal{I}_j) \geq \sqrt{k}$, since selling to Buyer 1 at price $\sqrt{k}\cdot 2^j$ already gives revenue of $\sqrt{k}$. Imagine an adversary who picks instance $\mathcal{I}_j$ with probability $\frac{1}{k}$ for each $j = 1, 2, \ldots, k$. A prior-independent mechanism cannot distinguish which instance the adversary is picking, and it is essentially run on the mixed instance $\mathcal{I^*}$ where
\begin{itemize}
\item $v_1 = \sqrt{k} \cdot 2^j$ with probability $2^{-j} \cdot \frac{1}{k}$, for each $j = 1, 2, \ldots, k$; and $v_1 = 1$ with the rest probability.
\item $v_2 = 1$ with probability $1$.
\end{itemize}
However, even the optimal auction has $\myr(\mathcal{I}^*) = \max_p \left(p \cdot \Pr[v_1 \geq p] + v_2 \cdot (1 - \Pr[v_1 \geq p])\right) \leq 1 + \frac{2}{\sqrt{k}}$. Therefore, the revenue of any prior-independent mechanism on this mixed instance $\mathcal{I}^*$ is at most $1 + \frac{2}{\sqrt{k}}$ too. This means for any prior-independent mechanism, its revenue is at most $1 + \frac{2}{\sqrt{k}}$ for some $\mathcal{I}_j$ where $j \in \{1, 2, \ldots, n\}$. This is only $O\left(\frac{1}{\sqrt{k}}\right) \cdot \myr$ or $\left(1 + O\left(\frac{1}{\sqrt{k}}\right)\right) \cdot \spa$ when $k \to +\infty$.
\end{proof}

This motivates us to limit the class of value distributions. We find that \emph{regularity}, which is a widely imposed assumption in auction theory, suffices for us to show the positive results. Some common examples of regular distributions are uniform, exponential, equal-revenue, and MHR ones.

Now we move on to show that having randomness in our mechanism is also mandatory for a non-trivial guarantee. We state this result in \cref{thm:no_randomness}.% for single-point (i.e. deterministic) distributions.

\begin{theorem}
For any constant $\varepsilon > 0$, no deterministic DSIC mechanism can guarantee a revenue of $\min(\varepsilon \cdot \myr, (1 + \varepsilon) \cdot \spa)$ for single-point valuation distributions, even when there are only two buyers.
\label{thm:no_randomness}
\end{theorem}
\begin{proof}
Suppose for the purpose of contradiction that such a deterministic truthful auction exists. It has to always allocate the item to someone, otherwise it gets $0$ revenue on that value pair and thus does not meet the theorem condition there.

For $M > 1 + \frac{1}{\varepsilon}$, it has to allocate to Buyer 2 on value pair $(v_1, v_2) = (1, M)$, and the payment has to be at least $(1 + \varepsilon)$.

Therefore, on value pair $(1, 1 + 0.5\varepsilon)$, it has to allocate to 
Buyer 1. Otherwise Buyer 2 will deviate to report $\hat v_2 = (1 + 0.5\varepsilon)$ on value pair $(1, M)$, to still get the item with less payment.

By the same logic, it has to allocate to Buyer 2 on value pair $\left(1, \frac{1}{1 + 0.5\varepsilon}\right)$. Therefore the allocation rule is not monotone -- Buyer 2 loses the allocation when his value increases -- and thus cannot be truthful.
\end{proof}

\section{Main Result: the Single-Item Case}
\label{sec:main-result}
We discuss our main result in this section. We give a family of parameterized prior-independent mechanisms that, informally speaking, always achieve a revenue (in expectation) either at least a certain fraction of the optimal prior-dependent expected revenue of $\myr$, or much better than the expected revenue of $\spa$. Our result holds for any set of buyers whose values are drawn from independent, but not necessarily identical, regular distributions.

We first define our geometric-threshold mechanisms, and then present our main result in \cref{thm:main}.
\begin{definition}[Geometric-Threshold Mechanisms]
Given parameters $\alpha\geq 1$ and $k\in \mathbb{Z}_{+}$, we denote $\rvcg(\alpha, k)$ as the threshold mechanism with the following $k + 1$ thresholds: $\lambda_1 = 1$ and $w_1 = \frac{1}{2}$; $\lambda_i = \alpha^{\frac{i - 1}{k}}$ and $w_i = \frac{1}{2k}$ for $i = 2, \ldots, k + 1$. Moreover, given a set of distributions on the private values of buyers, we abuse the notation to denote $\rvcg(\alpha, k)$ also as the expected revenue of the seller using the corresponding mechanism. 
\end{definition}
% \dw{Shall we stick with the name rVCG? We are looking at single item, and the pricing also doesn't have much explicit VCG connection.}
% \textcolor{red}{A: We should not use the term rVCG, but have some new name. e.g., in the new WWW paper on PoA of rVCG, I call it $RAND_\alpha$.}

\begin{theorem}
For any parameter $\tau > e$, there exists some $\alpha = O(\tau\ln^2\tau)$ and $k = \ln\alpha \in \mathbb{Z}_{+}$ such that for $n$ buyers  with private values drawn from independent (but not necessarily identical) regular distributions, 
\[
\rvcg(\alpha, k) \geq \min\left(\Omega\Big(\frac{1}{\ln\tau}\Big) \cdot \myr, \ \tau \cdot \spa\right).
\]
Recall $\rvcg(\alpha, k)$, $\myr$ and $\spa$ denote the expected revenue achieved by the respective mechanisms from selling one item to the $n$ buyers. 
\label{thm:main}
\end{theorem}

\cref{cor:main} is immediately implied by \cref{thm:main}, and captures our main message in a simpler form -- for any constant $\tau$, there is a prior-independent mechanism that either beats $\spa$ by a factor of $\tau$, or constant-approximates $\myr$.
\begin{corollary}
For any parameter $\tau = O(1)$, there exists some $\alpha$ and $k$ such that for $n$ buyers with private values drawn from independent (but not necessarily identical) regular distributions, 
\[
\rvcg(\alpha, k) \geq \min\left(\Omega(1) \cdot \myr, \ \tau \cdot \spa\right).
\]
\label{cor:main}
\end{corollary}

Before proving \cref{thm:main}, we first present a lemma which bounds the revenue of $\myr$ in our setting. 
\begin{lemma}
For $n$ buyers with values $v_1, \ldots, v_n$ drawn from independent regular distributions $V_1,\ldots,V_n$,
recall that $s_1$ is the median of the distribution of $v^{(1)} = \max_{i = 1}^n v_i$. The expected revenue of $\myr$ for selling one item to the $n$ buyers satisfies
\[
\frac{1}{2} s_1 \leq \myr \leq (1 + 2 \ln 2) s_1.
\]
\label{lem:approximate_myerson}
\end{lemma}
\begin{proof}
For the first inequality, notice that sequentially posting a price of $s_1$ for every buyer sells the item with probability $\frac{1}{2}$. Therefore, $\myr \geq \frac{1}{2} s_1$, since the revenue-optimal mechanism $\myr$ gets revenue at least that of the sequential posted-price mechanism.

For the second inequality, consider the virtual welfare (which equals the revenue) achieved by the optimal mechanism. Let $p_i$ be the probability that buyer $i$ wins in the optimal mechanism and $z_i$ be the $(1 - p_i)$-th quantile of $V_i$, i.e., $\Pr[v_i \geq z_i] = p_i$. Then using the ex-ante relaxation, the virtual welfare from buyer $i$ is at most $p_i \cdot z_i$. If $z_i \leq s_1$, then $p_i \cdot z_i \leq p_i \cdot s_1$. Otherwise (i.e. if $z_i > s_1$), then $p_i \cdot z_i \leq 2 \cdot \Pr[v_i \geq s_1] \cdot s_1$ by regularity of $V_i$ and \cref{lem:regularity}. Therefore,
\[
\myr \leq s_1 \cdot \sum_{i = 1}^n \left(p_i + 2\Pr[v_i > s_1]\right) \leq s_1 \cdot \left(1 + 2 \sum_{i = 1}^n \Pr[v_i > s_1]\right) \leq (1 + 2 \ln 2) s_1.
\]
The last step uses the fact that $\sum_{i = 1}^n \Pr[v_i > s_1] > \ln 2$ would imply
\[
\Pr\left[\max_{i = 1}^n v_i > s_1\right] = 1 - \prod_{i = 1}^n (1 - \Pr[v_i > s_1]) \geq 1 - \exp\left(-\sum_{i = 1}^n \Pr[v_i > s_1]\right) > \frac{1}{2},
\]
which contradicts with the definition of $s_1$ as the median of the distribution of $\max_{i = 1}^n v_i$.
\end{proof}

We note that there is a line of work on approximating revenue using simple mechanisms such as anonymous pricing (i.e., sequential posted pricing with the same price); see e.g. \cite{alaei2019optimal,jin19pricing}. \cref{lem:approximate_myerson} gives a simple bound using $s_1$, where the constant factor should be improvable with techniques from aforementioned work. We now proceed to the proof of our main theorem of this section.

\begin{proof}[Proof of \cref{thm:main}]
Again we let $s_1$ be the median of the distribution of $\max_{i = 1}^n V_i$. Without loss of generality, we assume $V_1$ is the distribution that maximizes $\Pr_{v_i \sim V_i}[v_i \geq s_1]$ over $i\in[n]$. Furthermore, we define $u_2$ to be the median of the distribution of $\max_{i = 2}^n V_i$. It is straightforward from the definitions that $s_1\geq u_2$. Most of our proof works with a generic $\alpha\geq e$ that gives $k=\ln\alpha \in \mathbb{Z}_{+}$, and we pick the appropriate $\alpha$ to get the guarantees in terms of $\tau$ in the theorem statement at the end of our proof.

We consider the following three cases: (1) When no single value distribution frequently exceeds $s_1$; (2) When some value distribution frequently exceeds $s_1$, and $s_1$ and $u_2$ are relatively close; (3) When some value distribution frequently exceeds $s_1$, and $s_1 \gg u_2$.

\textbf{Case (1):} If $\Pr[v_1 \geq s_1] \leq \frac{1}{4}$, we will show $\rvcg(\alpha, k)$ is a constant approximation to $\myr$. The intuition is that the second highest value is at least $s_1$ with constant probability; and if there are at least two values exceeding $s_1$, we will gain a revenue of at least $\frac{s_1}{2}$, since our mechanism uses a threshold $\lambda_1 = 1$ with probability $w_1 = \frac{1}{2}$. Formally, we know
\[
\rvcg(\alpha, k) \geq \frac{s_1}{2} \cdot \Pr\left[v^{(2)} \geq s_1\right].
\]
We will show that $\Pr\left[v^{(2)} \geq s_1\right]$ is at least a constant in this case. Enumerating which two values are at least $s_1$, we have
\begin{align*}
\Pr\left[v^{(2)} \geq s_1\right] &\geq \sum_{1 \leq i < j \leq n} \Pr[v_i \geq s_1] \cdot \Pr[v_j \geq s_1] \cdot \Pr\left[\max_{t \neq i, j} v_t < s_1\right]\\
&\geq \sum_{1 \leq i < j \leq n} \Pr[v_i \geq s_1] \cdot \Pr[v_j \geq s_1] \cdot \Pr\left[\max_{1 \leq t \leq n} v_t < s_1\right]\\
&= \frac{1}{2} \cdot \sum_{1 \leq i < j \leq n} \Pr[v_i \geq s_1] \cdot \Pr[v_j \geq s_1]\\
&= \frac{1}{4} \cdot \sum_{1 \leq i \leq n} \Pr[v_i \geq s_1] \cdot \sum_{j \neq i} \Pr[v_j \geq s_1],
\end{align*}
where the second last step uses the definition of $s_1$ being the median of the (continuous) distribution of $v^{(1)}$.
Further, since $\sum_{1 \leq j \leq n} \Pr[v_j \geq s_1] \geq \Pr[v^{(1)} \geq s_1] \geq \frac{1}{2}$ and $\Pr[v_i \geq s_1] \leq \Pr[v_1 \geq s_1] \leq \frac{1}{4}$ by our assumption in this case, we have $\sum_{j \neq i} \Pr[v_j \geq s_1] \geq \frac{1}{2}-\frac{1}{4}=\frac{1}{4}$. Therefore,
\[
\Pr\left[v^{(2)} \geq s_1\right] \geq \frac{1}{4} \cdot \frac{1}{2} \cdot \frac{1}{4} = \frac{1}{32},
\]
and thus
\[
\rvcg(\alpha, k) \geq \frac{s_1}{64} \geq \frac{1}{64 (1 + 2 \ln 2)} \cdot \myr.
\]

\noindent\textbf{Case (2):} If  $\Pr[v_1 \geq s_1] > \frac{1}{4}$ and $s_1 \leq 12 \alpha u_2$, we will show $\rvcg(\alpha, k)$ is an $\Omega\left(\frac{1}{k}\right)$-approximation to $\myr$. Note that by our choice of $k=\ln\alpha$, the consecutive thresholds in our mechanism are separated by a constant factor of $\alpha^{\nicefrac{1}{k}}=e$. 

Observe that if $v_1 \geq s_1$ and $\max_{i = 2}^n v_i \geq u_2$, then $v^{(2)}$ (i.e. the second highest value) is at least $u_2$ and $v^{(1)}\geq s_1$. In this case, the thresholds in our mechanisms are $v^{(2)}, e\cdot v^{(2)}, e^2\cdot v^{(2)}, \ldots, \alpha v^{(2)}$. When $\frac{v^{(1)}}{v^{(2)}}\leq \alpha$, there exists a threshold setting the price to be at least $\frac{v^{(1)}}{e}\geq \frac{s_1}{e}$, and when $\frac{v^{(1)}}{v^{(2)}}> \alpha$, the largest threshold will set the price to be $\alpha v^{(2)}\geq \alpha u_2\geq s_1/12$. The mechanism will pick each threshold with probability (at least) $\frac{1}{2k}$, and thus $\rvcg(\alpha, k)$ gets at least a revenue of $\frac{1}{2k} \cdot \frac{s_1}{12}$ by just looking at when $v_1 \geq s_1$ and $\max_{i = 2}^n v_i \geq u_2$. This allows us to show
\begin{align*}
\rvcg(\alpha, k) &\geq \frac{1}{2k} \cdot \frac{s_1}{12} \cdot \Pr\left[v_1 \geq s_1 \land \max_{i = 2}^n v_i \geq u_2\right]\\
&= \frac{1}{2k} \cdot \frac{s_1}{12} \cdot \Pr\left[v_1 \geq s_1\right] \cdot \Pr\left[\max_{i = 2}^n v_i \geq u_2\right]\\
&\geq \frac{1}{2k} \cdot \frac{s_1}{12} \cdot \frac{1}{4} \cdot \frac{1}{2} = \frac{s_1}{192k} \geq \frac{1}{192(1 + 2 \ln 2)k} \cdot \myr.
\end{align*}

\noindent\textbf{Case (3):} Otherwise (i.e., $\Pr[v_1 \geq s_1] > \frac{1}{4}$ and $s_1 > 12 \alpha u_2$), we will show $\rvcg(\alpha, k) = \Omega(\frac{\alpha}{k\ln\alpha}) \cdot \spa$. Notice that if $v_1 \geq s_1$ and $u_2 \leq \max_{2 \leq i \leq n} v_i \leq \frac{s_1}{\alpha}$, then $v^{(1)}=v_1$ and $v^{(2)}=\max_{2 \leq i \leq n} v_i\leq \frac{v^{(1)}}{\alpha}$, which means our mechanism will have revenue at least $\frac{1}{2k} \cdot \alpha \cdot \max_{2 \leq i \leq n} v_i$ by using the threshold $\alpha\cdot v^{(2)}$ with probability $\frac{1}{2k}$. Therefore,
\begin{align*}
\rvcg(\alpha, k) &\geq \E\left[\frac{1}{2k} \cdot \alpha \cdot \max_{i = 2}^n v_i \ \middle| \ v_1 \geq s_1 \land u_2 \leq \max_{i = 2}^n v_i \leq \frac{s_1}{\alpha}\right] \cdot \Pr\left[v_1 \geq s_1 \land u_2 \leq \max_{i = 2}^n v_i \leq \frac{s_1}{\alpha}\right]\\
&= \frac{\alpha}{2k} \cdot \E\left[\max_{i = 2}^n v_i \ \middle| \ u_2 \leq \max_{i = 2}^n v_i \leq \frac{s_1}{\alpha}\right] \cdot \Pr\left[v_1 \geq s_1\right] \cdot \Pr\left[u_2 \leq \max_{i = 2}^n v_i \leq \frac{s_1}{\alpha}\right]\\
&\geq \frac{\alpha}{2k} \cdot \E\left[\max_{i = 2}^n v_i \ \middle| \ u_2 \leq \max_{i = 2}^n v_i \leq \frac{s_1}{\alpha}\right] \cdot \frac{1}{4} \cdot \left(\Pr\left[\max_{i = 2}^n v_i \geq u_2\right] - \Pr\left[\max_{i = 2}^n v_i > \frac{s_1}{\alpha}\right]\right)\\
&\geq \frac{\alpha}{2k} \cdot \E\left[\max_{i = 2}^n v_i \ \middle| \ u_2 \leq \max_{i = 2}^n v_i \leq \frac{s_1}{\alpha}\right] \cdot \frac{1}{4} \cdot \left(\frac{1}{2} - \frac{1 + 2 \ln 2}{12}\right)\\
&= \frac{(5 - 2 \ln 2) \alpha}{96k} \cdot \E\left[\max_{i = 2}^n v_i \ \middle| \ u_2 \leq \max_{i = 2}^n v_i \leq \frac{s_1}{\alpha}\right].
\end{align*}
The second-last step uses the fact that $\Pr\left[\max_{i = 2}^n v_i \geq 12u_2\right] \leq \frac{1 + 2 \ln 2}{12}$; If that doesn't hold, sequential posted pricing at $12 u_2$ on buyers $2, 3, \ldots, n$ would give revenue more than $(1 + 2 \ln 2) u_2$, which contradicts \cref{lem:approximate_myerson}.

Next, we give an upper bound of similar form for $\spa$. Let $\sec_{i = j}^n v_i$ denote the second largest value from the set $\{v_j, v_{j+1}, \ldots, v_n\}$. We have
\begin{align*}
\spa &= \E\left[\sec_{i = 1}^n v_i\right] = \int_{0}^{+\infty} \Pr\left[\sec_{i = 1}^n v_i \geq t\right] \d t.
\end{align*}
Evaluating the integral separately at $t \in [0, u_2)$, $t \in \left[u_2, \nicefrac{s_1}{\alpha}\right)$, $t \in \left[\nicefrac{s_1}{\alpha}, s_1\right)$, and $t \in \left[s_1, +\infty\right)$, we get
\begin{align*}
\spa &\leq \int_{0}^{u_2} \d t + \int_{u_2}^{\nicefrac{s_1}{\alpha}} \Pr\left[\max_{i = 2}^n v_i \geq t\right] \d t + \int_{\nicefrac{s_1}{\alpha}}^{s_1} \Pr\left[\max_{i = 2}^n v_i \geq t\right] \d t + \\
&\quad \quad \int_{s_1}^{+\infty} \left(\Pr\left[v_1 \geq t \land \max_{i = 2}^n v_i \geq t\right] + \Pr\left[\sec_{i = 2}^n v_i \geq t\right]\right) \d t\\
&= u_2 + \int_{u_2}^{\nicefrac{s_1}{\alpha}} \left(\Pr\left[t \leq \max_{i = 2}^n v_i \leq \nicefrac{s_1}{\alpha}\right] + \Pr\left[\max_{i = 2}^n v_i > \nicefrac{s_1}{\alpha}\right]\right) \d t + \int_{\nicefrac{s_1}{\alpha}}^{s_1} \Pr\left[\max_{i = 2}^n v_i \geq t\right] \d t + \\
&\quad \quad \int_{s_1}^{+\infty} \left(\Pr\left[v_1 \geq t \land \max_{i = 2}^n v_i \geq t\right] + \Pr\left[\sec_{i = 2}^n v_i \geq t\right]\right) \d t\\
&\leq u_2 + \int_{u_2}^{\nicefrac{s_1}{\alpha}} \left(\Pr\left[t \leq \max_{i = 2}^n v_i \leq \nicefrac{s_1}{\alpha}\right] + \frac{(1 + 2 \ln 2) u_2 \alpha}{s_1}\right) \d t + \int_{\nicefrac{s_1}{\alpha}}^{s_1} \frac{(1 + 2 \ln 2) u_2}{t} \d t + \\
&\quad \quad \int_{s_1}^{+\infty} \left(\frac{s_1}{t} \cdot \frac{(1 + 2 \ln 2) u_2}{t} + \Pr\left[\sec_{i = 2}^n v_i \geq t\right]\right) \d t.
\end{align*}
In the last step we used the fact that $\Pr[v_1 \geq t] \leq \frac{s_1}{t}$, since pricing at $t$ for Buyer $1$ should not give revenue more than $s_1$ by \cref{lem:regularity}; and $\Pr\left[\max_{i = 2}^n v_i \geq t\right] \leq \frac{(1 + 2 \ln 2) u_2}{t}$, since sequential posted pricing at $t$ for Buyer $2, 3, \ldots, n$ should not give revenue more than $(1 + 2 \ln 2) u_2$, which is an upper bound for the optimal revenue given by \cref{lem:approximate_myerson} applied to (only) buyers $2,\ldots,n$. To continue with our derivation, we have
\begin{align*}
\spa &\leq u_2 + \int_{u_2}^{\nicefrac{s_1}{\alpha}} \left(\Pr\left[t \leq \max_{i = 2}^n v_i \leq \nicefrac{s_1}{\alpha}\right] + \frac{(1 + 2 \ln 2) u_2 \alpha}{s_1}\right) \d t + \int_{\nicefrac{s_1}{\alpha}}^{s_1} \frac{(1 + 2 \ln 2) u_2}{t} \d t + \\
&\quad \quad \int_{s_1}^{+\infty} \left(\frac{(1 + 2 \ln 2) s_1u_2}{t^2} + \Pr\left[\sec_{i = 2}^n v_i \geq t\right]\right) \d t\\
&\leq u_2 + \int_{u_2}^{\nicefrac{s_1}{\alpha}} \left(\Pr\left[t \leq \max_{i = 2}^n v_i \leq \nicefrac{s_1}{\alpha} \ \middle| \ u_2 \leq \max_{i = 2}^n v_i \leq \nicefrac{s_1}{\alpha}\right]\right) \d t + (1 + 2 \ln 2) u_2 + (1 + 2 \ln 2) u_2 \ln \alpha + \\
&\quad \quad \frac{(1 + 2 \ln 2) s_1u_2}{s_1} + \int_{s_1}^{+\infty} \Pr\left[\sec_{i = 2}^n v_i \geq t\right] \d t\\
&\leq u_2 + \int_{u_2}^{\nicefrac{s_1}{\alpha}} \left(\Pr\left[t \leq \max_{i = 2}^n v_i \leq \nicefrac{s_1}{\alpha} \ \middle| \ u_2 \leq \max_{i = 2}^n v_i \leq \nicefrac{s_1}{\alpha}\right]\right) \d t + \\
&\quad \quad \int_{s_1}^{+\infty} \Pr\left[\sec_{i = 2}^n v_i \geq t\right] \d t + (1 + 2 \ln 2) (2 + \ln \alpha) u_2\\
&= \E\left[\max_{i = 2}^n v_i \ \middle| \ u_2 \leq \max_{i = 2}^n v_i \leq \nicefrac{s_1}{\alpha}\right] + \int_{s_1}^{+\infty} \Pr\left[\sec_{i = 2}^n v_i \geq t\right] \d t + ((3 + 4 \ln 2) + (1 + 2 \ln 2) \ln \alpha) u_2.
\end{align*}
Finally, notice that
\begin{align*}
\int_{s_1}^{+\infty} \Pr\left[\sec_{i = 2}^n v_i \geq t\right] \d t &\leq \int_{s_1}^{+\infty} \Pr\left[\max_{i = 2}^n v_i \geq t\right]^2 \d t\\
&\leq \int_{s_1}^{+\infty} \frac{((1 + 2 \ln 2) u_2)^2}{t^2} \d t\\
&= \frac{((1 + 2 \ln 2) u_2)^2}{s_1} < \frac{(1 + 2 \ln 2)^2}{12} \cdot u_2,
\end{align*}
where we once again used \cref{lem:approximate_myerson}.
Therefore,
\begin{align*}
\spa &\leq \E\left[\max_{i = 2}^n v_i \ \middle| \ u_2 \leq \max_{i = 2}^n v_i \leq \nicefrac{s_1}{\alpha}\right] + \frac{(1 + 2 \ln 2)^2}{12} \cdot u_2 + ((3 + 4 \ln 2) + (1 + 2 \ln 2) \ln \alpha) u_2\\
&\leq \E\left[\max_{i = 2}^n v_i \ \middle| \ u_2 \leq \max_{i = 2}^n v_i \leq \nicefrac{s_1}{\alpha}\right] \cdot \left(1 + \frac{(1 + 2 \ln 2)^2}{12} + (3 + 4 \ln 2) + (1 + 2 \ln 2) \ln \alpha\right).
\end{align*}
Thus,
\[
\frac{\rvcg(\alpha, k)}{\spa} \geq \frac{\frac{(5 - 2 \ln 2) \alpha}{96k}}{\left(1 + \frac{(1 + 2 \ln 2)^2}{12} + (3 + 4 \ln 2) + (1 + 2 \ln 2) \ln \alpha\right)} \geq \frac{1}{256} \cdot \frac{\alpha}{k \ln \alpha},
\]%(5 - 2 * ln 2) / 96 / (1 + (1 + 2 * ln 2)^2 / 12 + (3 + 4 * ln 2) + (1 + 2 * ln 2))
when $\alpha \geq e$.

Taking $k=\ln\alpha$ and thus $\alpha^{1/k}=e$, we get that in all cases, $\rvcg$ is either at least $\frac{\alpha}{256\ln^2\alpha}\cdot \spa$ (i.e. Case (3)) or $\Omega(1/\ln\alpha)\cdot\myr$ (i.e. Cases (1),(2)). 

To get the guarantees in the theorem and corollary statements, when $\tau\in (1, e]$, it suffices to take $\alpha=e^{12}$ to get $\frac{\alpha}{256\ln^2\alpha} > e \geq \tau$ on the SPA side, and $\Omega(1/\ln\alpha)$ is $\Omega(1)$ on the Myerson side. When $\tau>e$, it suffices to take $\alpha \in [e^{11},e^{12}]\cdot\tau\ln^2\tau$ (with $k\in\mathbb{Z}_{+}$). It is easy to check $\frac{\alpha}{256\ln^2\alpha} \geq \tau$ on the SPA side, and $\Omega(1/\ln\alpha)$ is $\Omega(1/
\ln\tau)$ on the Myerson side.
\end{proof}

\subsection{Lower Bounds}
We complement our positive result with \cref{thm:lower}, which states that the guarantee in \cref{thm:main} is tight up to constants, even for $2$ buyers with deterministic value distributions.
\begin{theorem}
For $2$ buyers with deterministic values and any $\tau \geq 3$, there is no prior-independent DSIC mechanism $\mec$ with revenue satisfying
\[
\mec \geq \min\left(\frac{2.5}{\ln \tau} \cdot \myr, \ \tau \cdot \spa\right).
\]
\label{thm:lower}
\end{theorem}
\begin{proof}
Suppose for the purpose of contradiction that such a mechanism $\mec$ exists. Fix a parameter $m\in \mathbb{Z}_{+}$ to be decided later based on $\tau$, and consider the family of $m$ examples each with two buyers whose values are $v_1 = 1$ and $v_2(k) = 2^k$ for $k \in \{1, 2, \ldots, m\}$. Let $\mec_k$, $\myr_k$ and $\spa_k$ be the respective revenues from the point-distribution instance $(v_1, v_2(k))$. Treating each example as its own point-distribution, we need to satisfy the guarantee $\mec_k \geq \min\left(\frac{2.5}{\ln \tau} \cdot \myr_k, \ \tau \cdot \spa_k\right)$ for $(v_1, v_2(k))$ for all $k \in \{1, 2, \ldots, m\}$ simultaneously. Note that $\myr_k=2^k$ and $\spa_k = 1$ for all $k$.

We argue that $\mec$ cannot achieve this guarantee, by first proving $\sum_{k = 1}^m p_k \cdot \mec_k \leq 3$ where $p_k := \frac{1}{2^k} \cdot \frac{2^m}{2^m - 1}$ gives a probability distribution over $k$. To show this we consider a single instance with distribution over a support on the aforementioned $m$ examples. Suppose the instance $(v_1, v_2(k))$ appears with probability $p_k$ for each $k$, then $\mec$ can get a revenue of at most $3$ (at most $1$ from Buyer 1 and at most $2$ from Buyer 2 at any price) on this randomized instance. Choose $m$ so that $2^m + 1 \leq \tau < 2^{m + 1} + 1$, and thus $\mec_k \leq 2^m < \tau \cdot \spa_k$. This means for $\mec$ to exist we must have $\mec_k \geq \frac{2.5}{\ln \tau} \cdot \myr_k = \frac{2.5}{\ln \tau} \cdot 2^k$ for all $k$, then
\begin{align*}
\sum_{k = 1}^m p_k \cdot \mec_k &\geq \frac{2.5}{\ln \tau} \cdot \sum_{k = 1}^m p_k \cdot \myr_k\\
&= \frac{2.5}{\ln \tau} \cdot \sum_{k = 1}^m \frac{1}{2^k} \cdot \frac{2^m}{2^m - 1} \cdot 2^k\\
&= \frac{2.5m}{\ln \tau} \cdot \frac{2^m}{2^m - 1}
\ > 3.
\end{align*}
The contradiction implies the theorem statement.
\end{proof}

We note that constructions in the work of \cite{allouah2020prior} can give lower bounds when $\tau$ is close to $1$. We present \cref{thm:ab_lower1} and \cref{thm:ab_lower2} in addition to our lower bound of \cref{thm:lower}.

\cref{thm:ab_lower1} follows from the same instance as in \cite{allouah2020prior}, where $\spa$ is the optimal prior-independent auction there. Therefore, their lower bound that no prior-independent mechanism (with a technical assumption) can beat $0.715 \cdot \myr$ in the instance implies \cref{thm:ab_lower1}.

\begin{theorem} [\cite{allouah2020prior}]
\label{thm:ab_lower1}
Even for two i.i.d.\@ MHR distributions, for any $\varepsilon > 0$, no prior-independent DSIC mechanism $\mec$ with finite Arzel\`a variation (see \cite[Section 8]{allouah2020prior}) can always satisfy $\mec \geq \min(0.715 \cdot \myr, (1 + \varepsilon) \cdot \spa)$ .
\end{theorem}

For \cref{thm:ab_lower2}, we again look at the same family of instances as in \cite{allouah2020prior}. However, here we need to take into account the performance of $\spa$ and balance the parameters. They show no prior-independent mechanism can beat $0.556 \cdot \myr$, and later \cite{DBLP:conf/focs/HartlineJL20} give a stronger, tight impossibility result that no prior-independent mechanism can beat $0.524 \cdot \myr$.
We show no prior-independent mechanism can beat $\min(0.572 \cdot \myr, (1 + \varepsilon) \cdot \spa)$.

\begin{theorem}
\label{thm:ab_lower2}
Even for two i.i.d.\@ regular distributions, for any $\varepsilon > 0$, no prior-independent DSIC mechanism $\mec$ with finite Arzel\`a variation can always satisfy $\mec \geq \min((\frac{4}{7} + \varepsilon) \cdot \myr, (1 + \varepsilon) \cdot \spa)$.
\end{theorem}
\begin{proof}
Following the work of \cite{allouah2020prior}, we look at the regular distribution:
\[
F_a(v) =
  \begin{cases}
    1 - \frac{1}{v + 1}       & \quad \text{if } v < a\\
    1  & \quad \text{if } v \geq a
  \end{cases}.
\]
On a pair of these distributions, $\spa = \int_{0}^a \left(\Pr_{v \sim F_a} [v \geq t]\right)^2 \mathrm{d} t = \int_{0}^a \left(\frac{1}{t + 1}\right)^2 \mathrm{d} t = \frac{a}{a + 1}$, and $\myr = a \cdot (1 - (\frac{a}{a + 1})^2) = \frac{a(2a + 1)}{(a + 1)^2}$. \cite{allouah2020prior} show that for any $\varepsilon > 0$, no prior-independent mechanism $\mec$ can guarantee
\begin{align*}
\frac{\mec}{\myr} \geq (1 + \varepsilon) &\cdot \max\bigg(\frac{1}{2 - q}, \\
&\max_{\gamma > 1} \left(\frac{q}{2 - q} + 2 \cdot \frac{\gamma}{\gamma - 1} \cdot \frac{1}{1 - q} \cdot \frac{1}{2 - q} \cdot \left(\frac{1 - q}{1 - q + \gamma q} - \frac{1}{\gamma - 1} \ln \frac{\gamma}{1 - q + \gamma q}\right)\right)\bigg)
\end{align*}
where $q = \frac{1}{1 + a}$.

Let $a = 3$. We have $q = \frac{1}{4}$, and it is impossible to beat $\frac{4}{7} \myr = \spa$.
% 1/7 + 32/21*x/(x-1) * (3/(3+x) - 1/(x-1)*ln(4*x/(3+x)))
\end{proof}

\section{Selling Multiple Identical Items}
\label{sec:multi-item}
\label{SEC:MULTI-ITEM}
In this section, we consider a generalization where the seller is selling $k$ identical items. The buyers are unit-demand, meaning each of them can only receive at most one copy. Without loss of generality, we assume $\log_2 k$ is an integer to simply our exposition. (In general, we can reduce $k$ to the nearest power of $2$, and the loss of constant factors are absorbed in the theorem statements in this section.)

\paragraph{Multi-Item Threshold Mechanisms} We generalize the threshold mechanisms in \cref{sec:threshold} to the multi-item case. Such a mechanism with capacity parameter $t$ still uses a finite number of thresholds $\{\lambda_1, \lambda_2, \ldots, \lambda_m\}$, where $\lambda_i$ happens with probability $w_i$ with $\sum_{i = 1}^m w_i = 1$. For a value profile $(v_1, v_2, \ldots, v_n)$, the mechanism generates a random threshold $\lambda_i$ according to the probabilities $(w_1, w_2, \ldots, w_m)$. It then looks at values $v^{(j)}$ and $v^{(t + 1)}$, separately for each $j \in [t]$. If $v^{(j)} \geq \lambda_i \cdot v^{(t + 1)}$, one item is allocated to buyer $j$ for a price of $\lambda_i \cdot v^{(t + 1)}$. Otherwise, no item is allocated to buyer $j$.

\begin{definition}[Multi-Item Geometric-Threshold Mechanisms]
The mechanism $\mgtm(\tau)$ runs $\rvcg_{2^j}(\alpha, c)$ for $j \in \{0, 1, \ldots, \log_2 k\}$ each with probability $\frac{1}{1 + \log_2 k}$, where $\alpha = O\left(k^2\tau\ln^2(k\tau)\right)$ and $c = \ln\alpha \in \mathbb{Z}_{+}$. $\rvcg_{2^j}(\alpha, c)$ is the multi-item threshold mechanism with capacity parameter $t = 2^j$ and $m + 1$ thresholds: $\lambda_1 = 1$ and $w_1 = \frac{1}{2}$; $\lambda_i = \alpha^{\frac{i - 1}{k}}$ and $w_i = \frac{1}{2m}$ for $i = 2, \ldots, m + 1$.
\end{definition}
\begin{theorem}
For $n$ buyers and $k$ items, the mechanism $\mgtm(\tau)$ satisfies
\[
\mgtm(\tau) \geq \min\left(\Omega\Big(\frac{1}{\ln(k\tau)} \cdot \frac{1}{\ln k}\Big) \cdot \myr, \ \tau \cdot \vcg\right).
\]
\label{thm:multi_main}
\end{theorem}
\subsection{An Upper Bound for \texorpdfstring{$\myr$}{Myerson}}
To prove \cref{thm:multi_main}, we first give an upper bound for $\myr$. Recall we denote $s_i$ as the median of the distribution of $v^{(i)}$.
\begin{lemma}
$\myr \leq 12 \cdot \sum_{j = 0}^{\log_2 k} 2^j \cdot s_{2^j}$.
\label{lem:multi_myr}
\end{lemma}
\begin{proof}
The statement clearly follows from \cref{lem:approximate_myerson} when $k = 1$, and thus we assume $k \geq 2$ in the rest of the proof.

Suppose we have the following $(\log_2 k) + 1$ buckets $[0, s_{2^{\log_2 k}}], (s_{2^{\log_2 k}}, s_{2^{(\log_2 k) - 1}}], \ldots, (s_2, s_1]$. Bucket $j$ is the one whose largest value is $s_{2^j}$, where $j \in \{0, 1, \ldots, \log_2 k\}$. Put each value distribution $V_i$ into one of the buckets depending on which range the median of $V_i$ is in. (Note that the median of $V_i$ cannot exceed $s_1$, since $s_1$ is the median of the distribution of $v^{(1)}$.) Let $B_j$ be the set of indices of distributions in Bucket $j$, and let $n_j := |B_j|$.

Now we describe a relaxation of the allocation constraint. Instead of allocating to at most $k$ buyers, we allow allocation to any buyer in $B_0 \cup B_1 \cup \cdots \cup B_{(\log_2 k) - 1}$ arbitrarily, and require that we allocate to at most $k$ buyers in $B_{\log_2 k}$. The maximum revenue under these new relaxed constraints is an upper bound for $\myr$.

The revenue we can get from buyers in $B_0 \cup B_1 \cup \cdots \cup B_{(\log_2 k) - 1}$ is at most
\[
\sum_{j = 0}^{(\log_2 k) - 1} n_j \cdot s_{2^j}
\]
by \cref{lem:regularity}.

The virtual welfare, which is equal to the revenue, that we can get from buyers in $B_{\log_2 k}$ can be upper-bounded similar to the proof of \cref{lem:approximate_myerson}. Let $p_i$ be the allocation probability to buyer $i$, and let $w_i$ be the $(1 - p_i)$-th quantile of $V_i$, i.e., $\Pr[v_i \geq w_i] = p_i$. The virtual welfare from buyer $i$ is at most $p_i \cdot w_i$. If $w_i \leq s_k$, then $p_i \cdot w_i \leq p_i \cdot s_k$. Otherwise, $p_i \cdot w_i \leq 2 \cdot \Pr[v_i \geq s_k] \cdot s_k$ by \cref{lem:regularity}. To summarize, we have
\begin{align*}
\myr &\leq \left(\sum_{j = 0}^{(\log_2 k) - 1} n_j \cdot s_{2^j}\right) + \left(\sum_{i \in B_{\log_2 k}} p_i \cdot s_k\right) + \left(\sum_{i \in B_{\log_2 k}} 2 \cdot \Pr[v_i \geq s_k] \cdot s_k\right)\\
&\leq \left(\sum_{j = 0}^{(\log_2 k) - 1} n_j \cdot s_{2^j}\right) + k \cdot s_k + \left(\sum_{i \in B_{\log_2 k}} 2 \cdot \Pr[v_i \geq s_k] \cdot s_k\right).
\end{align*}
Further, using \cref{lem:multi_myr_1} and \cref{lem:multi_myr_2} that we are about to prove, we get
\begin{align*}
\myr &\leq \left(\sum_{j = 0}^{(\log_2 k) - 1} 12 \cdot 2^j \cdot s_{2^j}\right) + k \cdot s_k + 6 \cdot k \cdot s_k\\
&\leq 12 \cdot \sum_{j = 0}^{\log_2 k} 2^j \cdot s_{2^j},
\end{align*}
and thus finishes the proof.
\end{proof}
\begin{lemma}
For $j < \log_2 k$, we have $n_j \leq 12 \cdot 2^j$.
\label{lem:multi_myr_1}
\end{lemma}
\begin{proof}
Let $X_i = 1$ denote the event of $v_i \geq s_{2^{j + 1}}$ and $X_i = 0$ otherwise. $X_i$'s are mutually independent for $i \in B_j$. Let $\mu := \E\left[\sum_{i \in B_j} X_i\right]$. Using the Chernoff bound, we have
\[
\Pr\left[\sum_{i \in B_j} X_i \geq \frac{\mu}{2}\right] \geq 1 - e^{-\nicefrac{\mu}{8}}.
\]
Note that for each $i \in B_j$, $\Pr[v_i \geq s_{2^{j + 1}}] \geq \frac{1}{2}$. Therefore, $\mu \geq \frac{n_j}{2}$ and thus
\[
\Pr\left[\sum_{i \in B_j} X_i \geq \frac{n_j}{4}\right] \geq 1 - e^{-\nicefrac{n_j}{16}}.
\]
If $n_j > 12 \cdot 2^j$, then
\[
\Pr\left[\sum_{i \in B_j} X_i \geq 2^{j + 1}\right] \geq 1 - e^{-\nicefrac{12}{16}} > \frac{1}{2},
\]
contradicting with the definition of $s_{2^{j + 1}}$, as too frequently $2^{j + 1}$ values exceed $s_{2^{j + 1}}$.
\end{proof}

\begin{lemma}
If $k \geq 2$, then $\sum_{i \in B_{\log_2 k}} \Pr[v_i \geq s_k] \leq 3k$.
\label{lem:multi_myr_2}
\end{lemma}
\begin{proof}
Let $X_i = 1$ denote the event of $v_i \geq s_k$ and $X_i = 0$ otherwise. $X_i$'s are mutually independent for $i \in B_{\log_2 k}$. Let $\mu := \E\left[\sum_{i \in B_{\log_2 k}} X_i\right]$. Using the Chernoff bound, we have
\[
\Pr\left[\sum_{i \in B_{\log_2 k}} X_i \geq \frac{\mu}{2}\right] \geq 1 - e^{-\nicefrac{\mu}{8}}.
\]
If $\mu > 3k$, then 
\[
\Pr\left[\sum_{i \in B_{\log_2 k}} X_i \geq k\right] \geq 1 - e^{-\nicefrac{3k}{8}} > \frac{1}{2},
\]
contradicting with the definition of $s_k$, which is the median of the distribution of the $k$-th maximum value.
\end{proof}

\subsection{Extending the Single-Item Case}
\begin{lemma}
For $n$ buyers, $k$ items and parameter $\tau'$, there exists some $\alpha = O(\tau'\ln^2\tau')$ and $c = \ln\alpha \in \mathbb{Z}_{+}$ such that for each $t \in [k]$,
\[
\rvcg_t(\alpha, c) \geq \min\left(\Omega\Big(\frac{1}{\ln\tau'}\Big) \cdot t \cdot s_{t}, \ \tau' \cdot t \cdot \E[v^{(t + 1)}]\right).
\]
\label{thm:multi_core}
\end{lemma}
The proof of \cref{thm:multi_core} is conceptually similar to that of \cref{thm:main}, and therefore we postpone it to the appendix. That said, the generalization does require new technical insights. In particular, we prove new technical steps in the form of \cref{lem:anti-concentration} and \cref{lem:prob_bound}.

\cref{lem:anti-concentration} is the core lemma to generalize Case (1) of \cref{thm:main}. By definition of $s_t$, the number of values above $s_t$ is at least $t$ with probability $\frac{1}{2}$. What \cref{lem:anti-concentration} states is under the case condition, the number of values above $s_t$ is at least $t + 1$ with constant probability as well, in the style of an ``anti-concentration'' bound. \cref{lem:prob_bound} is a major step towards generalizing Case (3) of \cref{thm:main}. It is a tail upper bound for the minimum of multiple regular distributions, a counterpart of \cref{lem:regularity} in the single-item case.

\subsection{Completing the Proof}
\begin{proof}[Proof of \cref{thm:multi_main}]
Notice that $\mgtm(\tau)$ uses the mechanism $\rvcg_t(\alpha, c)$ with probability $\Omega\left(\frac{1}{\ln k}\right)$ for each $t=2^j$ where $j\in\{0,1,\ldots,\log_2 k\}$. For any $t$, we can use \cref{thm:multi_core} with $\tau' = k^2 \tau$ to get
\begin{align*}
\rvcg_t(\alpha, c) &\geq \min\left(\Omega\Big(\frac{1}{\ln\tau'}\Big) \cdot t \cdot s_{t}, \ \tau' \cdot t \cdot \E[v^{(t + 1)}]\right)\\
&\geq \min\left(\Omega\Big(\frac{1}{\ln(k^2\tau)}\Big) \cdot t \cdot s_{t}, \ k^2 \tau \cdot t \cdot \E[v^{(k + 1)}]\right)\\
&\geq \min\left(\Omega\Big(\frac{1}{\ln(k\tau)}\Big) \cdot t \cdot s_{t}, \ k \tau \cdot \vcg\right).
\end{align*}
The last step above uses $\vcg = k \cdot \E[v^{(k + 1)}]$.

Now we break into two cases:

\noindent\textbf{Case (1):} Suppose for some $j \in \{0, 1, \ldots, \log_2 k\}$, $\rvcg_{2^j}(\alpha, c) \geq k \tau \cdot \vcg$. Since $\mgtm(\tau)$ uses $\rvcg_{2^j}(\alpha, c)$ with probability $\frac{1}{1 + \log_2 k}$, we know $\mgtm(\tau)$ gets revenue at least $\tau \cdot \vcg$.

\noindent\textbf{Case (2):} Otherwise, we know for every $j \in \{0, 1, \ldots, \log_2 k\}$, $\rvcg_{2^j}(\alpha, c) \geq \Omega\Big(\frac{1}{\ln(k\tau)}\Big) \cdot 2^j \cdot s_{2^j}$. Additionally, \cref{lem:multi_myr} states that
\[
\sum_{j = 0}^{\log_2 k} 2^j \cdot s_{2^j} \geq \frac{1}{12} \cdot \myr.
\]
Therefore,
\begin{align*}
\mgtm(\tau) &\geq \Omega\Big(\frac{1}{\ln(k\tau)}\Big) \cdot \frac{1}{1 + \log_2 k} \sum_{j = 0}^{\log_2 k} 2^j \cdot s_{2^j}\\
&\geq \Omega\Big(\frac{1}{\ln(k\tau)} \cdot \frac{1}{\ln k}\Big) \cdot \myr.
\end{align*}
Combining the two cases gives the theorem statement.
\end{proof}

\section{Conclusions}
In this work, we studied the design of prior-independent auctions for bidders with heterogeneous value distributions. We showed a mechanism that can either achieve a constant fraction of the optimal revenue of any mechanism that knows the value distributions, or beat the revenue of the second-price auction by an arbitrarily large constant factor. Our mechanism has asymptotically optimal trade-off between the constants. We generalized our result to selling multiple identical items and gave a similar message. A possible future direction is to give better bounds and to consider further generalizations.

As another intriguing future direction, one can consider other ways to measure the effectiveness of prior-independent auctions for heterogeneous bidders. What does ``approximately optimal'' mean and how can we ``rank'' different mechanisms? We leave alternative answers to these questions for future work.

\bibliographystyle{alpha}
\bibliography{ref}

\newcommand{\etalchar}[1]{$^{#1}$}
\begin{thebibliography}{GHK{\etalchar{+}}06}

\bibitem[AB20]{allouah2020prior}
Amine Allouah and Omar Besbes.
\newblock Prior-independent optimal auctions.
\newblock {\em Management Science}, 66(10):4417--4432, 2020.

\bibitem[ABB22]{anunrojwong2022robustness}
Jerry Anunrojwong, Santiago~R. Balseiro, and Omar Besbes.
\newblock On the robustness of second-price auctions in prior-independent mechanism design.
\newblock In {\em Proceedings of the 23rd {ACM} Conference on Economics and Computation}, pages 151--152. {ACM}, 2022.

\bibitem[AHN{\etalchar{+}}19]{alaei2019optimal}
Saeed Alaei, Jason~D. Hartline, Rad Niazadeh, Emmanouil Pountourakis, and Yang Yuan.
\newblock Optimal auctions vs. anonymous pricing.
\newblock {\em Games and Economic Behavior}, 118:494--510, 2019.

\bibitem[AKW14]{DBLP:conf/soda/AzarKW14}
Pablo~Daniel Azar, Robert Kleinberg, and S.~Matthew Weinberg.
\newblock Prophet inequalities with limited information.
\newblock In {\em Proceedings of the 25th Annual {ACM-SIAM} Symposium on Discrete Algorithms}, pages 1358--1377, 2014.

\bibitem[BK96]{bulow1996auctions}
Jeremy Bulow and Paul Klemperer.
\newblock Auctions versus negotiations.
\newblock {\em American Economic Review}, 86(1):180--194, 1996.

\bibitem[BSW21]{braverman2021prior}
Mark Braverman, Jon Schneider, and S.~Matthew Weinberg.
\newblock Prior-free dynamic mechanism design with limited liability.
\newblock In {\em Proceedings of the 22nd ACM Conference on Economics and Computation}, pages 204--223, 2021.

\bibitem[BTC22]{bachrach2022distributional}
Nir Bachrach and Inbal Talgam-Cohen.
\newblock Distributional robustness: From pricing to auctions.
\newblock In {\em Proceedings of the 23rd ACM Conference on Economics and Computation}, page 150. Association for Computing Machinery, 2022.

\bibitem[BW19]{beyhaghi19competition}
Hedyeh Beyhaghi and S.~Matthew Weinberg.
\newblock Optimal (and benchmark-optimal) competition complexity for additive buyers over independent items.
\newblock In {\em Proceedings of the 51st Annual ACM SIGACT Symposium on Theory of Computing}, page 686–696, 2019.

\bibitem[CDFS19]{DBLP:conf/ec/CorreaDFS19}
Jos{\'{e}}~R. Correa, Paul D{\"{u}}tting, Felix~A. Fischer, and Kevin Schewior.
\newblock Prophet inequalities for {I.I.D.} random variables from an unknown distribution.
\newblock In {\em Proceedings of the 2019 {ACM} Conference on Economics and Computation}, pages 3--17, 2019.

\bibitem[CGL14]{DBLP:conf/stoc/ChenGL14}
Ning Chen, Nick Gravin, and Pinyan Lu.
\newblock Optimal competitive auctions.
\newblock In David~B. Shmoys, editor, {\em Proceedings of the 46th Annual Symposium on Theory of Computing}, pages 253--262. {ACM}, 2014.

\bibitem[Che22]{che2019distributionally}
Ethan Che.
\newblock Robustly optimal auction design under mean constraints.
\newblock In {\em Proceedings of the 23rd ACM Conference on Economics and Computation}, page 153–181. Association for Computing Machinery, 2022.

\bibitem[CR14]{cole2014sample}
Richard Cole and Tim Roughgarden.
\newblock The sample complexity of revenue maximization.
\newblock In {\em Proceedings of the 46th annual ACM symposium on Theory of computing}, pages 243--252, 2014.

\bibitem[CS21]{DBLP:conf/sigecom/CaiS21}
Linda Cai and Raghuvansh~R. Saxena.
\newblock 99{\%} revenue with constant enhanced competition.
\newblock In {\em Proceedings of the 22nd {ACM} Conference on Economics and Computation}, pages 224--241. {ACM}, 2021.

\bibitem[DRY15]{DBLP:journals/geb/DhangwatnotaiRY15}
Peerapong Dhangwatnotai, Tim Roughgarden, and Qiqi Yan.
\newblock Revenue maximization with a single sample.
\newblock {\em Games and Economic Behavior}, 91:318--333, 2015.

\bibitem[DSS19]{deng2019prior}
Yuan Deng, Jon Schneider, and Balasubramanian Sivan.
\newblock Prior-free dynamic auctions with low regret buyers.
\newblock In {\em Proceedings of the 2019 Conference on Neural Information Processing Systems}, pages 4804--4814, 2019.

\bibitem[EFF{\etalchar{+}}17]{eden2016competition}
Alon Eden, Michal Feldman, Ophir Friedler, Inbal Talgam{-}Cohen, and S.~Matthew Weinberg.
\newblock The competition complexity of auctions: {A} bulow-klemperer result for multi-dimensional bidders.
\newblock In {\em Proceedings of the 2017 {ACM} Conference on Economics and Computation}, page 343. {ACM}, 2017.

\bibitem[FFR18]{feldman201899}
Michal Feldman, Ophir Friedler, and Aviad Rubinstein.
\newblock 99\% revenue via enhanced competition.
\newblock In {\em Proceedings of the 2018 ACM Conference on Economics and Computation}, pages 443--460, 2018.

\bibitem[FH18]{DBLP:conf/focs/FengH18}
Yiding Feng and Jason~D. Hartline.
\newblock An end-to-end argument in mechanism design (prior-independent auctions for budgeted agents).
\newblock In {\em Proceedings of the 59th {IEEE} Annual Symposium on Foundations of Computer Science}, pages 404--415, 2018.

\bibitem[FHL21]{DBLP:conf/stoc/FengHL21}
Yiding Feng, Jason~D. Hartline, and Yingkai Li.
\newblock Revelation gap for pricing from samples.
\newblock In {\em Proceedings of the 53rd Annual {ACM} {SIGACT} Symposium on Theory of Computing}, pages 1438--1451, 2021.

\bibitem[FILS15]{fu2015randomization}
Hu~Fu, Nicole Immorlica, Brendan Lucier, and Philipp Strack.
\newblock Randomization beats second price as a prior-independent auction.
\newblock In {\em Proceedings of the 16th ACM Conference on Economics and Computation}, pages 323--323, 2015.

\bibitem[FLR19]{fu2019vickrey}
Hu~Fu, Christopher Liaw, and Sikander Randhawa.
\newblock The vickrey auction with a single duplicate bidder approximates the optimal revenue.
\newblock In {\em Proceedings of the 2019 ACM Conference on Economics and Computation}, pages 419--420, 2019.

\bibitem[GHK{\etalchar{+}}06]{DBLP:journals/geb/GoldbergHKSW06}
Andrew~V. Goldberg, Jason~D. Hartline, Anna~R. Karlin, Michael~E. Saks, and Andrew Wright.
\newblock Competitive auctions.
\newblock {\em Games and Economic Behavior}, 55(2):242--269, 2006.

\bibitem[GHW01]{DBLP:conf/soda/GoldbergHW01}
Andrew~V. Goldberg, Jason~D. Hartline, and Andrew Wright.
\newblock Competitive auctions and digital goods.
\newblock In {\em Proceedings of the 12th Annual Symposium on Discrete Algorithms}, pages 735--744. {ACM/SIAM}, 2001.

\bibitem[GHZ19]{DBLP:conf/stoc/GuoHZ19}
Chenghao Guo, Zhiyi Huang, and Xinzhi Zhang.
\newblock Settling the sample complexity of single-parameter revenue maximization.
\newblock In {\em Proceedings of the 51st Annual {ACM} {SIGACT} Symposium on Theory of Computing}, pages 662--673. {ACM}, 2019.

\bibitem[GLMN21]{golrezai2021boosted}
Negin Golrezaei, Max Lin, Vahab~S. Mirrokni, and Hamid Nazerzadeh.
\newblock Boosted second price auctions: Revenue optimization for heterogeneous bidders.
\newblock In {\em Proceedings of the 27th {ACM} {SIGKDD} Conference on Knowledge Discovery and Data Mining}, pages 447--457. {ACM}, 2021.

\bibitem[HJ21]{hartline2021lower}
Jason~D. Hartline and Aleck Johnsen.
\newblock Lower bounds for prior independent algorithms.
\newblock {\em arXiv preprint arXiv:2107.04754}, 2021.

\bibitem[HJL20]{DBLP:conf/focs/HartlineJL20}
Jason~D. Hartline, Aleck~C. Johnsen, and Yingkai Li.
\newblock Benchmark design and prior-independent optimization.
\newblock In {\em Proceedings of the 61st {IEEE} Annual Symposium on Foundations of Computer Science}, pages 294--305, 2020.

\bibitem[HR09]{HR2009simple}
Jason~D. Hartline and Tim Roughgarden.
\newblock Simple versus optimal mechanisms.
\newblock In {\em Proceedings of the 10th ACM Conference on Electronic Commerce}, pages 225--234, 2009.

\bibitem[JLQ{\etalchar{+}}19]{jin19pricing}
Yaonan Jin, Pinyan Lu, Qi~Qi, Zhihao~Gavin Tang, and Tao Xiao.
\newblock Tight approximation ratio of anonymous pricing.
\newblock In {\em Proceedings of the 51st Annual ACM SIGACT Symposium on Theory of Computing}, page 674–685, 2019.

\bibitem[KL03]{kleinberg2003value}
Robert Kleinberg and Tom Leighton.
\newblock The value of knowing a demand curve: Bounds on regret for online posted-price auctions.
\newblock In {\em Proceedings of the 44th Annual IEEE Symposium on Foundations of Computer Science}, pages 594--605. IEEE, 2003.

\bibitem[LMP23]{DBLP:conf/www/LiawMP23}
Christopher Liaw, Aranyak Mehta, and Andr{\'{e}}s Perlroth.
\newblock Efficiency of non-truthful auctions in auto-bidding: The power of randomization.
\newblock In {\em Proceedings of the 2023 {ACM} Web Conference}, pages 3561--3571. {ACM}, 2023.

\bibitem[Meh22]{Mehta22}
Aranyak Mehta.
\newblock Auction design in an auto-bidding setting: Randomization improves efficiency beyond {VCG}.
\newblock In {\em Proceedings of the 2022 {ACM} Web Conference}, pages 173--181. {ACM}, 2022.

\bibitem[Mye81]{myerson1981optimal}
Roger~B. Myerson.
\newblock Optimal auction design.
\newblock {\em Mathematics of Operations Research}, 6(1):58--73, 1981.

\bibitem[Nee03]{neeman2003effectiveness}
Zvika Neeman.
\newblock The effectiveness of english auctions.
\newblock {\em Games and Economic Behavior}, 43(2):214--238, 2003.

\bibitem[SS13]{sivan2013vickrey}
Balasubramanian Sivan and Vasilis Syrgkanis.
\newblock Vickrey auctions for irregular distributions.
\newblock In {\em Proceedings of the 9th International Conference on Web and Internet Economics}, pages 422--435. Springer, 2013.

\bibitem[Vic61]{vickrey1961counterspeculation}
William Vickrey.
\newblock Counterspeculation, auctions, and competitive sealed tenders.
\newblock {\em The Journal of Finance}, 16(1):8--37, 1961.

\bibitem[Wil89]{wilson1989game}
Robert Wilson.
\newblock Game theoretic analysis of trading.
\newblock In {\em Advances in Economic Theory: Fifth World Congress}, pages 33--70. CUP Archive, 1989.

\end{thebibliography}

\appendix
\section{Missing Proofs in \texorpdfstring{\cref{sec:multi-item}}{Section 5}}
\begin{lemma}[Anti-Concentration]
\label{lem:anti-concentration}
Let $X_1, \ldots, X_n$ be mutually independent Bernoulli random variables. $k < n$ is a positive integer. Suppose we have all three conditions below:
\begin{enumerate}
    \item $\E[X_1] \geq \E[X_2] \geq \cdots \geq \E[X_n]$.
    \item $\Pr\left[\sum_{i = 1}^n X_i \geq k\right] = \frac{1}{2}$.
    \item $\Pr\left[X_1 = X_2 = \cdots = X_k = 1\right] \leq \frac{1}{4}$.
\end{enumerate}
Then,
\[
\Pr\left[\sum_{i = 1}^n X_i \geq k + 1\right] > 0.01.
\]
\end{lemma}
\begin{proof}
We divide our proof into two cases:

\noindent\textbf{Case (1):} Suppose $\E[X_k] \leq \frac{1}{4}$. Define $Q_j = 1$ if $\sum_{i = 1}^j X_i \geq k$ and $Q_j = 0$ otherwise. We notice that
\begin{enumerate}
    \item $\E[Q_k] \leq \frac{1}{4}$.
    \item $\E[Q_n] = \frac{1}{2}$.
    \item $0 \leq \E[Q_{j + 1}] - \E[Q_j] \leq \E[X_{j + 1}] \leq \frac{1}{4}$, for each $j \in \{k, k + 1, \ldots, n - 1\}$.
\end{enumerate}
In other words, the sequence $\E[Q_j]$ is increasing in $j$ but cannot have a jump over $\frac{1}{4}$ in a single step. Therefore, there is some $j^* \in \{k, k + 1, \ldots, n - 1\}$, so that $\E[Q_{j^*}] \in [\frac{1}{8}, \frac{3}{8}]$.

Since $\Pr\left[\sum_{i = 1}^n X_i \geq k\right] = \frac{1}{2}$ and $\Pr\left[\sum_{i = 1}^{j^*} X_i \geq k\right] \leq \frac{3}{8}$, we know
\[
\Pr\left[\sum_{i = j^* + 1}^{n} X_i \geq 1\right] \geq \frac{1}{2} - \frac{3}{8} = \frac{1}{8}.
\]
Further,
\[
\Pr\left[\sum_{i = 1}^n X_i \geq k + 1\right] \geq \Pr\left[\sum_{i = 1}^{j^*} X_i \geq k\right] \cdot \Pr\left[\sum_{i = j^* + 1}^n X_i \geq 1\right] \geq \frac{1}{64}.
\]

\noindent\textbf{Case (2):} Suppose $\E[X_k] > \frac{1}{4}$. Define $R_j = 1$ if $\sum_{i = 1}^j X_i = j$ and $R_j = 0$ otherwise. Notice that
\begin{enumerate}
    \item $\E[R_0] = 1$.
    \item $\E[R_k] \leq \frac{1}{4}$.
    \item $\frac{1}{4} < \E[X_{j + 1}] = \frac{\E[R_{j + 1}]}{\E[R_j]} \leq 1$, for each $j \in \{0, 1, \ldots, k - 1\}$.
\end{enumerate}
In other words, the sequence $\E[R_j]$ is decreasing in $j$ but cannot decrease by a factor more than $\frac{1}{4}$ in a single step. Therefore, there is some $j^* \in \{1, 2, \ldots, k\}$, so that $\E[R_{j^*}] \in [0.1, 0.4]$.

Since $\Pr\left[\sum_{i = 1}^n X_i \geq k\right] = \frac{1}{2}$ and $\Pr\left[\sum_{i = 1}^{j^*} X_i \geq j^*\right] \leq 0.4$, we know
\[
\Pr\left[\sum_{i = j^* + 1}^{n} X_i \geq k - j^* + 1\right] \geq \frac{1}{2} - 0.4 = 0.1.
\]
Therefore,
\[
\Pr\left[\sum_{i = 1}^n X_i \geq k + 1\right] \geq \Pr\left[\sum_{i = 1}^{j^*} X_i \geq j^*\right] \cdot \Pr\left[\sum_{i = j^* + 1}^n X_i \geq k - j^* + 1\right] \geq 0.01. \qedhere
\]
\end{proof}

\begin{proof}[Proof of \cref{thm:multi_core}]
Similar to the proof of \cref{thm:main}, we denote $v^{(t)}$ as the $t$-th largest value among $v_1,\ldots,v_n$, and let $s_t$ be the median of the distribution of $v^{(t)}$. Without loss of generality, we index the bidders in non-increasing order according to $\Pr_{v_i \sim V_i}[v_i \geq s_t]$ for $i\in[n]$. Furthermore, we define $u_{t+1}$ to be the median of the distribution of $\max_{i = t+1}^n v_i$. Again we work with a generic $\alpha\geq e$ that gives $c=\ln\alpha \in \mathbb{Z}_{+}$, and we pick the appropriate $\alpha$ to get the guarantees in terms of $\tau''$ in the theorem statement at the end of our proof.

We can naturally extend the three cases considered in the proof of \cref{thm:main} to multiple items as follows.

\textbf{Case (1):} When $\Pr\left[v_i\geq s_t \forall i\in[1,t]\right] \leq \frac{1}{4}$, we will show $\rvcg(\alpha, k)$ is at least a constant factor of $s_t$. Define $X_i$ as the indicator of event $v_i\geq s_t$ for $i=1,\ldots,n$. Note our ordering of bidders, together with the definition of $s_t$ and the assumption of this specific case, allows us to use \cref{lem:anti-concentration}, which gives
\[
\Pr\left[v^{(t+1)} \geq s_t\right]\geq 0.01,
\]
so with at least a constant probability, we will have at least $t+1$ values that are at least $s_t$. When this happens, we will gain a revenue of at least $\frac{t\cdot s_t}{2}$, since our mechanism uses a threshold $\lambda_1 = 1$ with probability $w_1 = \frac{1}{2}$. Consequently, we know
\[
\rvcg_t(\alpha, c) \geq \frac{t\cdot s_t}{2} \cdot \Pr\left[v^{(t+1)} \geq s_t\right]\geq 0.005\cdot t\cdot s_t.
\]
\noindent\textbf{Case (2):} If  $\Pr\left[v_i\geq s_t \forall i\in[1,t]\right] > \frac{1}{4}$ and $s_t \leq 12 \alpha u_{t+1}$, we will show $\rvcg_t(\alpha, c)$ is at least an $\Omega\left(\frac{1}{k}\right)$ factor of $s_t$. Note that by our choice of $c=\ln\alpha$, the consecutive thresholds in our mechanism are separated by a constant factor of $\alpha^{\nicefrac{1}{k}}=e$. 

Observe that when $v_i \geq s_t$ for all $i\in [1,t]$ and $\max_{j = t+1}^n v_j \geq u_{t+1}$, then $v^{(t+1)}$ is at least $\min(s_t,u_{t+1})$, and thus at least $\frac{s_t}{12\alpha}$ from the assumption in this case. 
When this happens, the thresholds in our mechanisms are $v^{(t+1)}, e\cdot v^{(t+1)}, e^2\cdot v^{(t+1)}, \ldots, \alpha v^{(t+1)}$. When $\frac{v^{(t)}}{v^{(t+1)}}\leq \alpha$, there exists a threshold setting the price to be at least $\frac{v^{(t)}}{e}\geq \frac{s_t}{e}$, and when $\frac{v^{(t)}}{v^{(t+1)}}> \alpha$, the largest threshold will set the price to be $\alpha v^{(t+1)}\geq  s_t/12$. The mechanism will pick each threshold with probability (at least) $\frac{1}{2k}$, and thus $\rvcg_t(\alpha, c)$ gets at least a revenue of $\frac{t}{2k} \cdot \frac{s_t}{12}$ (since we sell $t$ items) by just looking at when $v_i \geq s_t$ for all $i\in[1,t]$ and $\max_{j = t+1}^n v_j \geq u_{t+1}$. This allows us to show
\begin{align*}
\rvcg_t(\alpha, c) &\geq \frac{t}{2k} \cdot \frac{s_t}{12} \cdot \Pr\left[v_i\geq s_t \forall i\in[1,t] \land \max_{j = t+1}^n v_j \geq u_{t+1}\right]\\
&= \frac{t}{2k} \cdot \frac{s_t}{12} \cdot \Pr\left[v_i\geq s_t \forall i\in[1,t]\right] \cdot \Pr\left[\max_{j = t+1}^n v_j \geq u_{t+1}\right]\\
&\geq \frac{t}{2k} \cdot \frac{s_t}{12} \cdot \frac{1}{4} \cdot \frac{1}{2} = \frac{s_t}{192k},
\end{align*}
where in the last step we used the assumption of this case and the definition of $u_{t+1}$ as the median of $\max_{j = t+1}^n v_j$.

\noindent\textbf{Case (3):} Otherwise (i.e., $\Pr\left[v_i\geq s_t \forall i\in[1,t]\right] > \frac{1}{4}$ and $s_t > 12 \alpha u_{t+1}$)
%, we will show $\rvcg_t(\alpha, c) = \Omega(\frac{\alpha}{k\ln\alpha}) \cdot \spa$
. Notice that if $v_i \geq s_t \forall i\in[1,t]$ and $u_{t+1} \leq \max_{t+1 \leq j \leq n} v_j \leq \frac{s_t}{\alpha}$, then $v^{(t)}\geq s_t$ and $v^{(t+1)}=\max_{t+1 \leq j \leq n} v_j\leq \frac{v^{(t)}}{\alpha}$, which means our mechanism will have revenue at least $\frac{t}{2k} \cdot \alpha \cdot \max_{t+1 \leq j \leq n} v_j$ by using the threshold $\alpha\cdot v^{(t+1)}$ with probability $\frac{1}{2k}$. Therefore,
\begin{align*}
\rvcg_t(\alpha, c) \geq &\E\left[\frac{t}{2k} \cdot \alpha \cdot \max_{j=t+1}^n v_j \ \middle| \ v_i \geq s_t \forall i\in[1,t] \land u_{t+1} \leq \max_{j = t+1}^n v_j \leq \frac{s_t}{\alpha}\right] \\
&\cdot \Pr\left[v_i \geq s_t\forall i\in[1,t] \land u_{t+1} \leq \max_{j = t+1}^n v_j \leq \frac{s_t}{\alpha}\right]\\
=& \frac{\alpha t}{2k} \cdot \E\left[ \max_{j=t+1}^n v_j\ \middle| \  u_{t+1} \leq \max_{j = t+1}^n v_j \leq \frac{s_t}{\alpha}\right] \cdot \Pr\left[v_i \geq s_t\forall i\in[1,t]\right] \cdot \Pr\left[u_{t+1} \leq \max_{j = t+1}^n v_j \leq \frac{s_t}{\alpha}\right]\\
\geq & \frac{\alpha t}{2k} \cdot \E\left[ \max_{j=t+1}^n v_j\ \middle| \ u_{t+1} \leq \max_{j = t+1}^n v_j \leq \frac{s_t}{\alpha}\right]  \cdot \frac{1}{4} \cdot \left(\Pr\left[\max_{j = t+1}^n v_j \geq u_{t+1}\right] - \Pr\left[\max_{j = t+1}^n v_j > \frac{s_t}{\alpha}\right]\right)\\
\geq & \frac{\alpha t}{2k} \cdot \E\left[ \max_{j=t+1}^n v_j\ \middle| \  u_{t+1} \leq \max_{j = t+1}^n v_j \leq \frac{s_t}{\alpha}\right]  \cdot \frac{1}{4} \cdot \left(\frac{1}{2} - \frac{1 + 2 \ln 2}{12}\right)\\
=& \frac{(5 - 2 \ln 2) \alpha t}{96k} \cdot \E\left[ \max_{j=t+1}^n v_j\ \middle| \ \max_{j = t+1}^n v_j \in[u_{t+1} , \frac{s_t}{\alpha}]\right] .
\end{align*}
The third-last step uses the assumption of this case, and the second-last step uses the definition of $u_{t+1}$ as the median, and also fact that $\Pr\left[\max_{j=t+1}^n v_j \geq 12u_{t+1}\right] \leq \frac{1 + 2 \ln 2}{12}$; If that doesn't hold, in the case of selling a single item to buyers $t+1,\ldots,n$, the sequential posted pricing at $12 u_{t+1}$ would give revenue more than $(1 + 2 \ln 2) u_{t+1}$, which contradicts \cref{lem:approximate_myerson}.

Next, we give an upper bound of similar form for $\E\left[v^{(t+1)}\right]$. Again, we use $\sec_{i = j}^n v_i$ to denote the second largest value from the set $\{v_j, v_{j+1}, \ldots, v_n\}$. We have
\begin{align*}
\E\left[v^{(t+1)}\right] = \int_{0}^{+\infty} \Pr\left[v^{(t+1)} \geq x\right] \d x.
\end{align*}
Evaluating the integral separately at $x \in [0, u_{t+1})$, $x \in \left[u_{t+1}, \nicefrac{s_t}{\alpha}\right)$, $x \in \left[\nicefrac{s_t}{\alpha}, s_t\right)$, and $x \in \left[s_t, +\infty\right)$, we get
\begin{align*}
\E\left[v^{(t+1)}\right] &\leq \int_{0}^{u_{t+1}} \d x + \int_{u_{t+1}}^{\nicefrac{s_t}{\alpha}} \Pr\left[\max_{j = t+1}^n v_j \geq x\right] \d x + \int_{\nicefrac{s_t}{\alpha}}^{s_t} \Pr\left[\max_{j = t+1}^n v_j \geq x\right] \d x + \int_{s_t}^{+\infty} \Pr\left[v^{(t+1)} \geq x \right] \d x\\
&\leq u_{t+1} + \int_{u_{t+1}}^{\nicefrac{s_t}{\alpha}} \left(\Pr\left[\max_{j = t+1}^n v_j \in[x, \nicefrac{s_t}{\alpha}]\right] + \Pr\left[\max_{j = t+1}^n v_j > \nicefrac{s_t}{\alpha}\right]\right) \d x + \int_{\nicefrac{s_t}{\alpha}}^{s_t} \Pr\left[\max_{j = t+1}^n v_j \geq x\right] \d x + \\
&\quad \quad \int_{s_t}^{+\infty} \left(\Pr\left[v_i \geq x\forall i\in[1,t] \land \max_{j = t+1}^n v_j \geq x\right] + \Pr\left[\sec_{j = t+1}^n v_j \geq x\right]\right) \d x\\
&\leq u_{t+1} + \int_{u_{t+1}}^{\nicefrac{s_t}{\alpha}} \left(\Pr\left[ \max_{j = t+1}^n v_j \in[x,\nicefrac{s_t}{\alpha}]\right] + \frac{(1 + 2 \ln 2) u_{t+1} \alpha}{s_t}\right) \d x + \int_{\nicefrac{s_t}{\alpha}}^{s_t} \frac{(1 + 2 \ln 2) u_{t+1}}{x} \d x + \\
&\quad \quad \int_{s_t}^{+\infty} \left(\frac{3s_t}{x} \cdot \frac{(1 + 2 \ln 2) u_{t+1}}{x} + \Pr\left[\sec_{j=t+1}^n v_j \geq x\right]\right) \d x.
\end{align*}
The first step uses the fact that we always have $v^{(t+1)}\leq \max_{j = t+1}^n v_j$. In the last step we used \cref{lem:prob_bound} and $\frac{8}{e}\leq 3$; and $\Pr\left[\max_{j = t+1}^n v_j \geq x\right] \leq \frac{(1 + 2 \ln 2) u_{t+1}}{x}$, since sequential posted pricing at $x$ for Buyer $t+1, \ldots, n$ should not give revenue more than $(1 + 2 \ln 2) u_{t+1}$, which is an upper bound for the optimal revenue given by \cref{lem:approximate_myerson} applied to selling one item to buyers $t+1,\ldots,n$. To continue with our derivation, we have
\begin{align*}
\E\left[v^{(t+1)}\right] &\leq u_{t+1} + \int_{u_{t+1}}^{\nicefrac{s_t}{\alpha}} \left(\Pr\left[ \max_{j = t+1}^n v_j \in[x,\nicefrac{s_t}{\alpha}]\right] + \frac{(1 + 2 \ln 2) u_{t+1} \alpha}{s_t}\right) \d x + \int_{\nicefrac{s_t}{\alpha}}^{s_t} \frac{(1 + 2 \ln 2) u_{t+1}}{x} \d x \\
&\quad \quad + \int_{s_t}^{+\infty} \left(\frac{3(1 + 2 \ln 2)s_t u_{t+1}}{x^2} + \Pr\left[\sec_{j=t+1}^n v_j \geq x\right]\right) \d x.\\
&\leq u_{t+1} + \int_{u_{t+1}}^{\nicefrac{s_t}{\alpha}} \left(\Pr\left[\max_{j = t+1}^n v_j \in[x, \nicefrac{s_t}{\alpha}] \ \middle| \ \max_{j = t+1}^n v_j \in[u_{t+1}, \nicefrac{s_t}{\alpha}]\right]\right) \d x + (1 + 2 \ln 2) u_{t+1}\\
&\quad \quad + (1 + 2 \ln 2) u_{t+1} \ln \alpha +  \frac{3(1 + 2 \ln 2) s_tu_{t+1}}{s_t} + \int_{s_t}^{+\infty} \Pr\left[\sec_{j = t+1}^n v_j \geq x\right] \d x\\
&\leq u_{t+1} + \int_{u_{t+1}}^{\nicefrac{s_t}{\alpha}} \left(\Pr\left[\max_{j = t+1}^n v_j \in[x, \nicefrac{s_t}{\alpha}] \ \middle| \ \max_{j = t+1}^n v_j \in[u_{t+1}, \nicefrac{s_t}{\alpha}]\right]\right)\d x \\
&\quad \quad + \int_{s_t}^{+\infty} \Pr\left[\sec_{j = t+1}^n v_j \geq x\right] \d x + (1 + 2 \ln 2) (4 + \ln \alpha) u_{t+1}\\
&= \E\left[\max_{j = t+1}^n v_j \ \middle| \  \max_{j = t+1}^n v_j \in[u_{t+1}, \nicefrac{s_t}{\alpha}]\right] + \int_{s_t}^{+\infty} \Pr\left[\sec_{j = t+1}^n v_j \geq x\right] \d x\\
&\quad \quad + ((5 + 8 \ln 2) + (1 + 2 \ln 2) \ln \alpha) u_{t+1}.
\end{align*}
Finally, notice that
\begin{align*}
\int_{s_t}^{+\infty} \Pr\left[\sec_{j = t+1}^n v_j \geq x\right] \d x &\leq \int_{s_t}^{+\infty} \Pr\left[\max_{j = t+1}^n v_j \geq x\right]^2 \d x\\
&\leq \int_{s_t}^{+\infty} \frac{((1 + 2 \ln 2) u_{t+1})^2}{x^2} \d x\\
&= \frac{((1 + 2 \ln 2) u_{t+1})^2}{s_t} < \frac{(1 + 2 \ln 2)^2}{12} \cdot u_{t+1},
\end{align*}
where we once again used \cref{lem:approximate_myerson} in the second step, and our assumption $s_t > 12 \alpha u_{t+1}$ with $\alpha \geq e$ in the last step.

Therefore, 
\begin{align*}
& \E\left[v^{(t+1)}\right]\\ 
\leq & \E\left[\max_{j = t+1}^n v_j \ \middle| \ \max_{j = t+1}^n v_j \in[u_{t+1} ,\nicefrac{s_t}{\alpha}]\right] + \frac{(1 + 2 \ln 2)^2}{12} \cdot u_{t+1} + ((5 + 8 \ln 2) + (1 + 2 \ln 2) \ln \alpha) u_{t+1}\\
\leq & \E\left[\max_{j = t+1}^n v_j \ \middle| \ \max_{j = t+1}^n v_j \in[u_{t+1} ,\nicefrac{s_t}{\alpha}]\right] \cdot \left(1 + \frac{(1 + 2 \ln 2)^2}{12} + (5 + 8 \ln 2) + (1 + 2 \ln 2) \ln \alpha\right),
\end{align*}
where the last step is because the expectation is always at least $u_{t+1}$.
Thus,
\[
\frac{\rvcg_t(\alpha, c)}{\E\left[v^{(t+1)}\right]} \geq \frac{\frac{(5 - 2 \ln 2) \alpha t}{96k}}{\left(1 + \frac{(1 + 2 \ln 2)^2}{12} + (5 + 8 \ln 2) + (1 + 2 \ln 2) \ln \alpha\right)} \geq \frac{1}{512} \cdot \frac{\alpha t}{k \ln \alpha},
\]%(5 - 2 * ln 2) / 96 / (1 + (1 + 2 * ln 2)^2 / 12 + (3 + 4 * ln 2) + (1 + 2 * ln 2))
when $\alpha \geq e$.

Taking $c=\ln\alpha$ and thus $\alpha^{1/c}=e$, we get that in all cases, $\rvcg_t(\alpha,c)$ is either at least $\frac{\alpha}{512\ln^2\alpha t}\cdot \E\left[v^{(t+1)}\right]$ (i.e. Case (3)) or $\Omega(t/\ln\alpha)\cdot s_t$ (i.e. Cases (1),(2)). 

To get the guarantees in the theorem and corollary statements, when $\tau'\in (1, e]$, it suffices to take $\alpha=e^{20}$ to get $\frac{\alpha}{512\ln^2\alpha} > e >= \tau'$ on the $\E\left[v^{(t+1)}\right]$ side, and $\Omega(1/\ln\alpha)$ is $\Omega(1)$ on the $s_t$ side. When $\tau'>e$, it suffices to take $\alpha \in [e^{19},e^{20}]\cdot\tau'\ln^2\tau'$ (with $c\in\mathbb{Z}_{+}$). It is easy to check $\frac{\alpha}{512\ln^2\alpha} \geq \tau'$ on the $\E\left[v^{(t+1)}\right]$ side, and $\Omega(1/\ln\alpha)$ is $\Omega(1/
\ln\tau')$ on the $s_t$ side.
\end{proof}
Now we prove the lemma used in the above proof.
\begin{lemma}
\label{lem:prob_bound}
$\Pr[v_i \geq x \forall i \in [1, t]] \leq \frac{8}{e} \cdot \frac{s_t}{x}$ for $x > s_t$.
\end{lemma}
\begin{proof}
For a regular distribution $V_i$, the revenue is concave in probability of selling. Taking the probability of selling as $1$, $\Pr[v_i \geq s_t]$, and $\Pr[v_i \geq x]$, we have
\[
s_t \cdot \Pr[v_i \geq s_t] \geq x \cdot \Pr[v_i \geq x] \cdot \frac{1 - \Pr[v_i \geq s_t]}{1 - \Pr[v_i \geq x]} + 0.
\]
Rearranging and using $\Pr[v_i \geq s_t] \leq 1$, we have
\[
\Pr[v_i < x] \geq \frac{x}{s_t} \cdot \Pr[v_i \geq x] \cdot \Pr[v_i < s_t].
\]
Define $y$ so that $\min_{i = 1}^t \Pr[v_i \geq y] = \frac{1}{2}$.
Now we consider the following two cases:

\noindent\textbf{Case (1):} We have $x \leq y$. In this case,
\[
\Pr[v_i < x] \geq \frac{x}{2 s_t} \cdot \Pr[v_i < s_t].
\]
Therefore,
\begin{align*}
\Pr[v_i \geq x \forall i \in [1, t]] &= \prod_i \Pr[v_i \geq x]\\
&= \exp\left(\sum_i \ln (1 - \Pr[v_i < x])\right)\\
&\leq \exp\left(-\sum_i \Pr[v_i < x]\right)\\
&\leq \exp\left(-\sum_i \frac{x}{2 s_t} \cdot \Pr[v_i < s_t]\right).
\end{align*}

Note that $\sum_i \Pr[v_i < s_t] \geq \Pr[\min_{i = 1}^t v_i < s_t] \geq \frac{1}{2}$. We get
\[
\Pr[v_i \geq x \forall i \in [1, t]] \leq \exp\left(- \frac{x}{4 s_t}\right) \leq \frac{4s_t}{e x}.
\]

\noindent\textbf{Case (2):} Now consider the case where $x > y$. Let $i^*$ be the buyer with $\Pr[v_{i^*} \geq y] = \frac{1}{2}$. We have
\[
\Pr[v_{i^*} \geq x] \leq \frac{y}{x}
\]
by \cref{lem:regularity}. Therefore,
\begin{align*}
\Pr[v_i \geq x \forall i \in [1, t]] &\leq \Pr[v_i \geq y \forall i \in [1, t]] \cdot \Pr[v_{i^*} \geq x \mid v_{i^*} \geq y]\\
&\leq \frac{4s_t}{e y} \cdot \frac{2y}{x} < \frac{8 s_t}{e x}.
\end{align*}
Combining the two cases gives the lemma statement.
\end{proof}

\section{Characterizations of Optimal Prior-Independent Mechanisms}
\label{sec:characterization}

In this section, we give characterizations for the optimal prior-independent mechanisms for 2 buyers. Our proofs are inspired by and generalize those of \cite{allouah2020prior} for i.i.d\@. distributions.

\begin{definition} [Scale-Free Mechanisms \cite{allouah2020prior}]
A mechanism $\mec$ is \emph{scale-free} if
\[
x_i(\theta v_i, \theta v_{-i}) = x_i(v_i, v_{-i})
\]
for any $\theta > 0, v_i, v_{-i} \geq 0, i = 1, 2$.
\label{def:scale_free}
\end{definition}

\begin{lemma}
If the guarantee $\mec \geq \min(\alpha \cdot \myr, \beta \cdot \spa)$ is satisfied by a prior-independent mechanism $\mec$ with finite Arzel\`a variation, then the same guarantee can be satisfied by a scale-free prior-independent mechanism.
\label{lem:scale}
\end{lemma}
\cref{lem:scale} intuitively makes sense: If the mechanism does not know the prior distributions, then after any scaling, the instance should be essentially the same. The proof of the i.i.d\@. case in \cite{allouah2020prior} directly extends to this heterogeneous case.

\begin{definition} [Symmetric Mechanisms]
A mechanism $\mec$ is \emph{symmetric} if
\[
x_1(v_1 = v, v_2 = v') = x_2(v_1 = v', v_2 = v)
\]
for any $v, v' \geq 0$.
\label{def:symmetric}
\end{definition}

\begin{lemma}
If the guarantee $\mec \geq \min(\alpha \cdot \myr, \beta \cdot \spa)$ is satisfied by a scale-free prior-independent mechanism $\mec$, then the same guarantee can be satisfied by a scale-free symmetric prior-independent mechanism.
\label{lem:symmetric}
\end{lemma}
\begin{proof}
If a scale-free prior-independent mechanism $\mec$ satisfies the lemma condition, then another mechanism $\mec'$ constructed by switching the roles of Buyer 1 and Buyer 2 also satisfies the lemma condition. Randomizing between $\mec$ and $\mec'$ with equal probability is a scale-free symmetric prior-independent mechanism, and it also guarantees a revenue of $\min(\alpha \cdot \myr, \beta \cdot \spa)$.
\end{proof}

\begin{lemma} [Myerson's Lemma~\cite{myerson1981optimal}]
$p_i(v_i, v_{-i}) = v_i \cdot x_i(v_i, v_{-i}) - \int_0^{v_i} x_i(t, v_{-i}) \d t + p_i(0, v_{-i})$, where $x_i$ is the allocation function and $p_i$ is the payment function.
\label{lem:myerson}
\end{lemma}
Without loss of generality, a revenue-maximizing auction should set $p_i(0, v_{-i})$ to be $0$, as it cannot be positive by the individual rationality (IR) constraint. This allows us to derive the revenue for any mechanism using only the allocation functions. Let $\rev_1(F, v_2)$ be the expected revenue from Buyer 1 when $v_1$ is drawn from $F$ and $v_2$ is fixed. We have
\begin{align*}
\rev_1(F, v_2) &= \int_0^M p_1(v_1, v_2) \d F(v_1)\\
&= \int_0^M \left[v_1 \cdot x_1(v_1, v_2) - \int_0^{v_1} x_1(t, v_2) \d t\right] \d F(v_1)\\
&= \int_0^M \left[v_1 \cdot x_1(v_1 / v_2, 1) - \int_0^{v_1} x_1(t / v_2, 1) \d t\right] \d F(v_1)
\end{align*}
where the last step above uses \cref{lem:scale}.

Similarly, let $\rev_2(v_1, G)$ be the expected revenue from Buyer 2 when $v_1$ is fixed and $v_2$ is drawn from $G$. We have
\[
\rev_2(v_1, G) = \int_0^M \left[v_2 \cdot x_2(1, v_2 / v_1) - \int_0^{v_2} x_2(1, t / v_1) \d t\right] \d G(v_2).
\]

Using the symmetry condition of \cref{lem:symmetric}, we get
\[
\rev_2(v_1, G) = \int_0^M \left[v_2 \cdot x_1(v_2 / v_1, 1) - \int_0^{v_2} x_1(t / v_1, 1) \d t\right] \d G(v_2).
\]

Abbreviating $x_1(v, 1)$ as $x_1(v)$, we have
\begin{align*}
\rev &= \int_0^M \rev_1(F, v_2) \d G(v_2) + \int_0^M \rev_2(v_1, G) \d F(v_1)\\
&= \int_0^M \int_0^M \left[v_1 \cdot x_1(v_1 / v_2) - \int_0^{v_1} x_1(t / v_2) \d t\right] \d F(v_1) \d G(v_2) +\\
&\quad \quad \int_0^M \int_0^M \left[v_2 \cdot x_1(v_2 / v_1) - \int_0^{v_2} x_1(t / v_1) \d t\right] \d G(v_2) \d F(v_1)\\
&= \int_0^M \int_0^M \left[v_1 \cdot x_1(v_1 / v_2) + v_2 \cdot x_1(v_2 / v_1) - \int_0^{v_1} x_1(t / v_2) \d t - \int_0^{v_2} x_1(t / v_1) \d t\right] \d F(v_1) \d G(v_2).
\end{align*}

Let $x_1(r) = \sum_{k = 1}^n \frac{1}{n} \cdot \mathbf{1}[r \geq \gamma_k]$ with an even $n$, as an approximation to the possibly continuous increasing function of $x_1(r)$.\footnote{This approximation can be arbitrarily close similar to the proof in \cite{allouah2020prior}.} Assume the $\gamma_k$'s are decreasing. The constraint on $x_1(r)$ is: $x_1(r) + x_1(1/r) \leq 1, \,\forall r$, which is equivalent to $\gamma_k \cdot \gamma_{n + 1 - k} > 1, \forall k$.

We rewrite $\rev$ as:
\begin{align*}
\rev &= \sum_{k = 1}^n \frac{1}{n} \cdot \int_0^M \int_0^M \bigg[v_1 \cdot \mathbf{1}[v_1 / v_2 \geq \gamma_k] + v_2 \cdot \mathbf{1}[v_2 / v_1 \geq \gamma_k] - \\
&\quad \quad \int_0^{v_1} \mathbf{1}[t / v_2 \geq \gamma_k] \d t - \int_0^{v_2} \mathbf{1}[t / v_1 \geq \gamma_k] \d t\bigg] \d F(v_1) \d G(v_2).
\end{align*}

Note that one pair of $(\gamma_k, \gamma_{n + 1 - k})$ suffices to give optimal revenue (within the class of scale-free symmetric mechanisms) against any fixed instance.

\end{document}